\documentclass[journal,onecolumn]{IEEEtran}
\usepackage{graphicx,subfig,amsthm,amsmath,latexsym,amssymb,times,subfloat}
\usepackage{float,epsfig,multirow,rotating,times,verbatim,wrapfig}
\usepackage{color,array, algpseudocode}
\usepackage{cite}
\usepackage{units}
\usepackage{mathtools}
\usepackage{soul}

\usepackage{array}

\newtheorem{theorem}{Theorem}

\newtheorem{proposition}[theorem]{Proposition}
\newtheorem*{remark}{Remark}
\newtheorem*{fact}{Fact}

\definecolor{menucolor}{rgb}{0.1,0.52,0.47}
\definecolor{urlcolor}{rgb}{0.85,0.37,0.01}
\definecolor{runcolor}{rgb}{0.46,0.44,0.701}
\definecolor{filecolor}{rgb}{0.2,0.5,0.01}
\definecolor{linkcolor}{rgb}{0.12,0.47,0.70}
\definecolor{citecolor}{rgb}{0.55,0.36,0.01}
\definecolor{anchorcolor}{rgb}{0.4,0.4,0.4}
\usepackage[colorlinks=true, urlcolor=urlcolor, linkcolor=linkcolor,citecolor=citecolor,filecolor=filecolor, anchorcolor=anchorcolor,menucolor=menucolor]{hyperref}
\usepackage{nameref}
\usepackage{zref-xr}
\usepackage{cite}
\zxrsetup{
  tozreflabel=false,
  toltxlabel=true,  
}

\newcommand{\I}{\mathbb I}
\newcommand{\G}{\mathbb G}

\newcommand{\balpha}{\boldsymbol{\alpha}}

\newcommand{\bbeta}{\boldsymbol{\beta}}

\newcommand{\bD}{\boldsymbol{D}}

\newcommand{\bx}{\boldsymbol{x}}
\newcommand{\bX}{\boldsymbol{X}}

\newcommand{\bY}{\boldsymbol{Y}}

\newcommand{\mN}{\mathcal N}

\newcommand{\E}{\mbox{I}\!\mbox{E}}

\newcommand{\ben}{\begin{enumerate}}
\newcommand{\een}{\end{enumerate}}

\newcommand{\TB}[1]{\textcolor{blue}{#1}}

\makeatletter
\newcommand\code{\bgroup\@makeother\_\@makeother\~\@makeother\$\@codex}
\def\@codex#1{{\normalfont\ttfamily\hyphenchar\font=-1 #1}\egroup}
\makeatother

\usepackage{etoolbox}

\makeatletter
\newcommand{\changeoperator}[1]{%
  \csletcs{#1@saved}{#1@}%
  \csdef{#1@}{\changed@operator{#1}}%
}
\newcommand{\changed@operator}[1]{%
  \mathop{%
    \mathchoice{\textstyle\csuse{#1@saved}}
               {\csuse{#1@saved}}
               {\csuse{#1@saved}}
               {\csuse{#1@saved}}%
  }%
}
\makeatother

\changeoperator{sum}
\changeoperator{prod}

\newcommand{\simp}{\stackrel{\mbox{\scriptsize p}}{\sim}}
\newcommand{\simv}{\stackrel{\mbox{\scriptsize v}}{\sim}}

\newcommand{\mP}{\mathbb{P}}

\newcommand{\bit}{\begin{itemize}}
\newcommand{\eit}{\end{itemize}}
\newcommand{\beqn}{\begin{equation}}
\newcommand{\eeqn}{\end{equation}}
\newcommand{\bea}{\begin{eqnarray*}}
\newcommand{\eea}{\end{eqnarray*}}
\newcommand{\bpf}{\begin{proof}}
\newcommand{\epf}{\end{proof}\ms}
\newcommand{\ms}{\medskip}

\newcommand{\citep}{\cite}
\newcommand{\citet}{\cite}
\newcommand{\citeyear}{\cite}
\newcommand{\citeauthor}{\cite}

\begin{document}
\title{Approximations in the homogeneous Ising model
}
  \author{Alejandro~Murua-Sazo~and~Ranjan~Maitra
  \thanks{A. Murua-Sazo is with the Department of Mathematics and
    Statistics at the Universit\'e de Montr\'eal,
    Montr\'eal, Qu\'ebec,
    Canada.}
  \thanks{R.Maitra is with the Department of Statistics, Iowa State
     University, Ames, Iowa, USA.}
  \thanks{This research was supported in part by the
    National Institute of Biomedical Imaging and Bioengineering (NIBIB) of the National
Institutes of Health (NIH) under its Award Nos. R21EB016212, and
R21EB034184, by the Natural Sciences and Engineering Research  and by the
Natural Sciences and Engineering Research Council of Canada (NSERC)
through Grant Nos. 327689-06 and 2019-05444, and by the  United States
Department of Agriculture (USDA) National 
Institute of Food and Agriculture (NIFA) Hatch project IOW03717.
The content of this paper is however solely the responsibility of the
authors and does not represent the official views of the the NIBIB,
the NIH, the NIFA, the USDA or the NSERC.} 

}
\IEEEcompsoctitleabstractindextext{%
  \begin{abstract}
    The Ising model is important in statistical modeling and inference in many applications, however its normalizing constant, mean number of active vertices and mean spin interaction -- quantities needed in inference -- are computationally intractable. We provide accurate approximations that make it possible to numerically calculate these quantities in the homogeneous case. Simulation studies indicate good performance of our approximation formulae that are scalable and unfazed by the size (number of nodes, degree of graph) of the Markov Random Field.  The practical import of our approximation formulae is illustrated in performing Bayesian inference in a functional Magnetic Resonance Imaging activation detection experiment, and also in likelihood ratio testing for anisotropy in the spatial patterns of yearly increases in pistachio tree yields.

  \end{abstract}
}
\maketitle
\IEEEdisplaynotcompsoctitleabstractindextext


\section{Introduction}
\label{introduction}
Let $\bX = \{X_1,\ldots,X_n\}$ be $n$ binary random variables with a
conditional graph dependence structure specified via the neighborhood
${\mN} = \{(i,j):i\!\sim\! j, 1 \leq i < j\leq\ n\}$, where the
notation $i\!\sim\! j$ indicates existence of an edge between the $i$th and
$j$th nodes in the graph. The  celebrated Ising model  of statistical physics
specifies the joint probability mass function (PMF) 
\begin{equation}
  \begin{split}
\mP(\bX=\bx; \balpha,\bbeta)\propto    \exp\bigl[ 
  \sum_{i} \alpha_ix_i 
 + 
  \sum_{i \sim j} \beta_{ij}\bigl\{x_ix_j+(1-x_i)(1-x_j)\bigr\}
  \bigr],
  \end{split}
  \label{eq:Ising}
\end{equation}
where it is assumed that $x_i\in\{0,1\}\,\forall i=1,2,\ldots, n$,  
$\bx=(x_1,\ldots,x_n)$, $\balpha = (\alpha_1,\ldots,\alpha_n)$ and $\bbeta
= (\beta_{ij}; 1 \leq i < j\leq\ n)$. The parameter $\alpha_i\geq 0$
modulates the 
chance that $X_i\!=\!1$ while the  parameter $\beta_{ij}$ specifies the 
strength of the interaction between $X_i$ and $X_j$ 
when $i\!\sim\! j$. The PMF~\eqref{eq:Ising} has
summation constant or partition function denoted by $Z(\balpha,\bbeta)$.

Model~\eqref{eq:Ising} was proposed by~\citet{Lenz-1920} to his
student Ernst 
Ising as a way to characterize magnetic phase 
transitions or singularities in the partition function over a
lattice graph. \citet{Ising-1925} published the model that bears his
name and showed that in one dimension, that is,  for a linear lattice graph,
the phase transition structure is trivial with no singularities in
the partition function. The distribution has applicability in
disciplines beyond physics -- indeed, one of its earliest uses in the
statistical literature was as 
a prior model for a binary scene in image
analysis~\citep{Besag-1986}. Other applications include state-time 
disease surveillance \citep{Robertson-et-al-2010,Jarpe-1999} and
mapping  \citep{Lee-Mitchell-2012,Ma-Carlin-2007,Ma-et-al-2010};
modeling of protein hydrophobicity \citep{Irback-et-al-1996},
genetic codon bias thermodynamics \citep{Rowe-Trainor-1983},
DNA elasticity \citep{Ahsan-et-al-1998} or ion channel interaction
\citep{Liu-Dilger-1993} in statistical genetics; 
modeling of electrophysiological phenomena of the retina
\citep{Schneidman-et-al-2006} and cortical recordings in neuroscience
\citep{Yu-et-al-2013,Hamilton-et-al-2013,Ganmor-et-al-2011};
 and modeling of
biological evolution \citep{Baake-et-al-1997}. The Ising model has
also been used to model voting patterns of senators in the US
Congress~\citep{Banerjee-etal-2008} or behaviors on social
networks \citep{Wasserman-Pattison-1996,Klemm-et-al-2003}. While many of
these applications use a regular lattice structure, some~(e.g. \citep{Banerjee-etal-2008,Wasserman-Pattison-1996,Klemm-et-al-2003})
use more general non-lattice structures.

Parameter estimation in the Ising model is often challenging, 
especially in the context of multi-dimensional lattices, and more generally for
non-lattice conditional dependence graphs. This difficulty flows from the
computational impracticality of obtaining exact closed-form
expressions for the partition function in the
presence of 
$\beta_{ij}$s.
Indeed, the computation of $Z(\balpha,\bbeta)$ in such graphs with
an external field has been shown to be 
NP-complete~\citep{Barahona-1982}. 
Even though the partition function is just a finite sum of exponential 
functions and can consequently be analytically expressed so that there 
is no phase transition for the finite graph Ising model, 
its computation is still intractable.
Nevertheless, parameter estimation is needed in many applications,
as in the two illustrative examples of 
Section~\ref{application}. In some cases, parameter estimation has
been eschewed in favor of approaches that 
are not always wholly 
satisfactory but obviate the need for its estimation. 
For example, 
\citet{Ma-Carlin-2007}  would ideally have liked to have estimated the
interaction parameter for assessing  
Medicare service area boundaries for competing hospice systems in
Duluth, Minnesota, but they instead fixed the value for their study. Some
authors have used empirical techniques such as
pseudolikelihood~(e.g.,\citep{Besag-1986,Wasserman-Pattison-1996,Klemm-et-al-2003}),
or written \eqref{eq:Ising} in terms of an exponential family model
and then used moment-matching or maximum entropy
methods \citep{Maitra-Besag-1998,Yu-et-al-2013,Hamilton-et-al-2013,Ganmor-et-al-2011}. Yet
others~\citep{Majewski-et-al-2001,Ahsan-et-al-1998,Irback-et-al-1996,Liu-Dilger-1993,Rowe-Trainor-1983}
have used simpler models, typically restricting to first
order interactions,  in order to employ a recursive algorithm to estimate
the partition function~\citep{Reichl-1980} which is possible only with
an Ising model with only nearest-neighbor (NN) structure (equivalently, first-order
interactions). \citet{Mo-Liang-2010} used~\eqref{eq:Ising} in one
dimension and with first-order neighborhood to signal if a probe
(gene) is enriched or not, and specifically mentioned that they 
did not model more complex interactions because of the intractability of
the partition function. Thus, there is need for a general method for
estimating the partition function in several applications.
Many methods have been suggested to estimate the normalizing
constant of intractable PMFs and densities.
Some
authors~(e.g.~\citep{kaufman49,schraudolphandkamenetsky09,karandashevandmalsagov17})
provided exact formulae for $Z(\balpha,\bbeta)$ 
in a 2D planar graph that also mostly assume a null $\balpha$ in
\eqref{eq:Ising}. These calculations also grow
with the size of a graph and are not particularly useful for
situations ({\em e.g.}, Bayesian inference) where the partition function needs
repeated evaluation. The popular method of {path sampling}
\citep{Ogata-1989,Gelman-Meng-1998,Richardson-Green-1997} writes the
normalizing constant as a function of the integral of an
expectation, that is estimated by Monte Carlo or
Markov Chain Monte Carlo (MCMC) sampling. Further, 
path sampling is usually implemented as a preprocessing step by
evaluating $Z(\balpha,\bbeta)$ on a grid of parameter values. 

Numerous other stochastic approaches~(for a sampling, consider \citep{Wang-Landau-2001,Ferrenberg-Swendsen-1989,Kumar-et-al-1992,Atchade-Liu-2010,Atchade-et-al-2013,Liang-2005,DelMoral-et-al-2006,liang10,liangetal16,zhangandliang23,lyneetal15,parkandharan20,bolandetal18,geyer99}) exist,  
but they all come with major drawbacks beyond their need for long sampling periods in both burn-in and post-burn-in phases, that can be computationally demanding and are not scalable to larger problems.
Further, these stochastic approaches can only estimate the partition function or moments at a discrete grid of parameter values, with interpolation (of unclear accuracy) needed for intermediate values, since these values are needed in a continuum, for instance, in the case of maximum likelihood parameter estimation and inference. 
In this paper, we therefore provide numerical approximations for the partition function,  the mean number of active modes, and the mean spin interaction  of the Ising model. Section~\ref{approximations} provides these approximation formulae separately for the isotropic ({\em i.e.}, for when 
$\alpha_i\equiv \alpha$ and $\beta_{ij}\equiv\beta$ for $1\leq i \ne
j\leq n$) and anisotropic cases.
Our approximations yield formulae with computational expense unaffeced by the size (number of nodes, graph degree) of the Ising model, and therefore, unlike those obtained from stochastic methods, essentially infinitely scalable. 
Section~\ref{simulations} evaluates the accuracy of these approximation formulae.
Section~\ref{application} illustrates the utility and performance of these
approximations in the context of Bayesian and likelihood-based
inference. 
The paper concludes with some discussion. A supplement, with sections and equations prefixed by ``S'' is also available.


\section{Approximations in  large Ising models}
\label{approximations}
\subsection{The isotropic case}
\label{sec:isotropic}
Our starting point is~\eqref{eq:Ising} under the assumptions of
homogeneity and isotopy. Since 
$ x_jx_i {+}  (1{-}x_i)(1{-}x_j) = 1 {-} (x_i {-} x_j)^2$, the isotropic
homogeneous Ising model can be written as 
\begin{equation}
\mP(\bx;\alpha,\beta) {=} Z(\alpha,\beta)^{-1} \exp\{
\alpha \displaystyle\sum_{i=1}^n x_i {-} \beta \displaystyle\sum_{i \sim
  j} (x_i - x_j)^2 \},
\label{eq:pmf.iso}
\end{equation}
where we write $Z(\balpha,\bbeta)$ as $Z(\alpha,\beta)$ to reflect that $\balpha$ and $\bbeta$ can be characterized by scalars $\alpha$ and $\beta$ in the isotropic homogenous context being considered here.
\subsubsection{The normalizing constant}
\label{approx:Z}
Let $G_{n,m} = (V, E)$ be the graph underlying the data, where $m$
denotes the number of edges in the graph. The set of nodes is
$V=\{x_1, \ldots, x_n\}$ and the set of edges is $E=\{ (x_i, x_j) : i\sim j\}$.
Since each $x_i$ is either 0 or 1, we get the equivalent relation
\begin{equation*}
  \displaystyle\sum_{i \sim j} (x_i - x_j)^2
= \displaystyle\sum_{i=1}^n k_i x_i   -  2 \displaystyle\sum_{i\sim j} x_i x_j,
\end{equation*} 
where $k_i$ is the degree of the $i$th node, that is, the number of edges
linked to the node $i$, for $i=1,\ldots,n$. 
This article assumes a regular $G_{n,m}$~(see Ch. 3 of \citep{bollobas01}), that is, $k_i = k$ for all nodes.
The exponent in $\mP(\bx;\alpha,\beta)$  is then
$\alpha'  \displaystyle\sum_{i=1}^n x_i  + \beta \displaystyle\sum_{i,j} \eta_{ij} x_i x_j,$
where $\alpha'= \alpha - k\beta$ and $\eta_{ij}=\I(i\sim j)$, with
$\I(\cdot)$ the indicator function. 
Let $M(\ell)$ be the set of sequences $\bx$ 
with $ \displaystyle\sum_{i=1}^n x_i = \ell$. Then
\begin{equation}
\begin{split}
  Z(\alpha,\beta)
 =  1 &+ \exp{(\alpha' n + \beta k n)} + n\exp(\alpha')
+\displaystyle\sum_{\ell=2}^{n-1} \exp(\alpha' \ell)
\displaystyle\sum_{X \in M(\ell)} \exp\bigl(\beta \displaystyle\sum_{i, j} \eta_{ij} x_i x_j
\bigr). 
\end{split}
\label{eq:Z:ell}
\end{equation}
With a graph $G_{n,m}$ having a total of $m= nk/2$ edges and for a fixed $\ell$, {only  $\ell$ observations}, say  $\tilde{\bX}_{\ell}
= \{ x_{i_1}, 
x_{i_2},\ldots, x_{i_{\ell}}\}$, contribute to $\displaystyle\sum_{i, j} \eta_{ij} x_i x_j$, so that
$\displaystyle\sum_{i, j} \eta_{ij} x_i x_j = \displaystyle\sum_{h=1}^{\ell} \displaystyle\sum_{j}\eta_{i_h,j} x_j$.
The set $\tilde{\bX}_{\ell}$ can be thought of as a subgraph of the original graph,
namely $G_{\ell}=( \tilde{\bX}_{\ell}, \tilde{E}_{\ell}),$
where the set of edges is a subset of all possible $\ell_2=  \ell(\ell -1)/2$
edges between nodes in $\tilde{\bX}_{\ell}$ that are
present in $G_{n,m}$. In graph theory, the subgraphs $G_{\ell}$ are referred to as
node-induced subgraphs.
The node $x_t$ contributes to the sum $\displaystyle\sum_{h=1}^{\ell} \displaystyle\sum_{j}\eta_{i_h,j} x_j$ only if $x_t$ is a node of
$G_{\ell}.$  
This sum corresponds to twice the number of edges in $G_{\ell}$.
Computing the last sum in (\ref{eq:Z:ell})
corresponds to counting the number of subgraphs
$G_{\ell}$ with a given number of edges.
Our partition function approximation 
uses
\begin{proposition}\label{prop:Z}
  Let $Y(s,\ell)$ be the number of node-induced subgraphs $G_{\ell}$
  containing exactly $s$ edges.
  Then
  \begin{equation}
      Z(\alpha, \beta)  =  1  + \exp{(\alpha' n + \beta k n)} + n\exp(\alpha')
+\displaystyle\sum_{\ell=2}^{n-1} \binom{n}{\ell} \exp(\alpha' \ell) {\cal
  M}_{p(\cdot |\ell)}(2\beta),
\label{eq:partition:1}
\end{equation}
where $  {\cal M}_{p(\cdot |\ell)}(2\beta)$ is the moment generating
function (MGF) 
associated with the distribution $p(s |\ell) = Y(s,\ell) / \binom{n}{\ell},$
$s \in \{0,1,\ldots, \ell k/2\},$ evaluated at $t=2\beta.$
\end{proposition}

\begin{proof}
Write the last sum in (\ref{eq:Z:ell}) as
\begin{equation*}
  \begin{split}
    \displaystyle\sum_{X \in M(\ell)} \exp\bigl(\beta \displaystyle\sum_{i, j} \eta_{ij} x_i x_j
\bigr)  = \displaystyle\sum_{s =0}^{ \ell_2}
Y(s, \ell) \exp\bigl(2\beta s \bigr) = \binom{n}{\ell} \displaystyle\sum_{s =0}^{ \ell_2}
\exp\bigl(2\beta s \bigr)\, \tfrac{Y(s, \ell)}{\binom{n}{\ell}}.
\end{split}
\end{equation*}
Conditional of $M(\ell)$, the proportions $Y(s, \ell) /
\binom{n}{\ell}$ define a distribution on $s$. 
Because we suppose that the graph $G_{n,m}$ is regular, the support of $p(s|\ell)$ is over
$\{0,1,\ldots, \ell k/2\}$.
\end{proof}
The number of edges, $s$, present in a given node-induced subgraph is half
    the sum of the degrees associated with the nodes present in the subgraph.
    Since as shown below
    the sum of degrees is the sum of several sums,
we can approximate the distribution of $s$ given $M(\ell)$ for large $\ell$
via  a normal distribution, and replace ${\cal M}_{p(\cdot
  |\ell)}(\cdot)$ by the Gaussian MGF. We formalize this  approach
further next. 

For a given graph $G_{n,m}$, and $s$, $p(s|\ell)$
corresponds to the proportion of node-induced subgraphs
with $\ell$ nodes and exactly $s$ edges between the
nodes. Unfortunately, calculating $p(s|\ell)$ is not straightforward.
So, instead of the number of edges $s$, we consider 
the degree distribution of the subgraphs. 
Let $r_{\ell,h} = \displaystyle\sum_{j=1}^n \eta_{i_h, j} x_j$, for $h=1,2,\ldots, \ell.$
These quantities are the observed degrees of the nodes $x_{i_1},
\ldots, x_{i_\ell}$. 
Let $ r_{\ell} = \displaystyle\sum_{h=1}^{\ell} r_{\ell,h}$.
The quantities 
$r_{\ell,h}$ and $r_{\ell}$ are realizations from the distribution of
the degrees of the graph. 
In particular, finding their first two moments is enough as they 
fully characterize the normal distribution. The distribution of the number of node-induced subgraphs
with degrees adding up to $r_\ell$ 
has the form
$$p_d( r_{\ell} |\ell ) = \displaystyle\sum_{r_{\ell,1} + r_{\ell,2} + \cdots r_{\ell,\ell} = r_{\ell} } p( r_{\ell,1}, r_{\ell,2},\ldots,r_{\ell,\ell}).$$
The support of this distribution lies over the even numbers
$r_{\ell} = 2s$. Also, the joint PMF $p( r_{\ell,1}, r_{\ell,2},\ldots,r_{\ell,\ell})$ is  not
  straightforward to compute. However, the marginals are easily obtained for a regular  graph with $k$ edges for each node.
In this case, \textcolor{blue}{$r_{\ell,h}$,} the proportion of edges 
for a given node $x_{i_h}$ in a subgraph of $\ell$ nodes, has the
hypergeometric distribution with parameters $( n - 1, k, \ell - 1)$.
Therefore, the expectation of twice the number of edges is given by $\mu_{\ell}= E( r_{\ell}) = \ell E( r_{\ell,h}) = \ell (\ell-1) k/(n-1) = 2\ell_2 \theta,$
where $\theta = m/{n\choose2}= k/(n-1)$ is the proportion of edges
with respect to a complete graph. The variance depends on the
dependency between the $r_{\ell,h}$s. Proceeding, we have the following
\begin{proposition}
  \label{prop:var}
  Let $\sigma_{\ell}^2 = 2\ell_2\theta (1 - \theta) ( 1 - \tfrac{\ell - 2}{n-2} )$,
  and
$\rho_{\ell} = (\ell -1)(n-2k) /\{(n-2)(n-k-1)\}.$
We have $\operatorname{Var}(r_{\ell}  ) = \sigma_{\ell}^2 \bigl( 1 - \rho_{\ell} \bigr),$ and
$ \operatorname{Cov}( r_{\ell,t}, r_{\ell,h} ) =  - \sigma_{\ell}^2 \rho_{\ell}/(2\ell_2) = {\cal O}(n^{-1}).$
In particular, 
  for all $\ell = o(n)$,
$\operatorname{Var}(r_{\ell} ) / \sigma_{\ell}^2 \rightarrow 1$ as $n\rightarrow \infty$,
or equivalently $\rho_{\ell}\rightarrow 0$ as $n\rightarrow \infty$, uniformly on $\ell = o(n)$.
\end{proposition}

\begin{proof} 
We have $ \operatorname{Var}(\sum_{h=1}^{\ell} r_{\ell,h})
{=} \sum_{h=1}^{\ell} (\ell-1)\theta (1 - \theta) ( 1 - y_{\ell,2}) 
+ 2\sum_{t < h} \operatorname{Cov}( r_{\ell,t}, r_{\ell,h})
= \sigma_{\ell}^2 + 2\sum_{t < h} \operatorname{Cov}( r_{\ell,t}, r_{\ell,h}),$
with $y_{\ell,2}= (\ell - 2)/(n-2)$, and $\sigma_{\ell}^2 = 2\ell_2\theta (1 - \theta) ( 1 - y_{\ell,2}).$
Let $k_{th}$ be the number of neighbors in common between vertices $t$ and $h$.
In the notation that follows, the conditional expectation given $\eta$, means 
that the values of the couples $(t,h)$ are fixed. 
Further, 
\begin{multline*}
\operatorname{Cov}( r_{\ell,t}, r_{\ell,h} | \eta ) = \displaystyle\sum_{i=1}^n
\displaystyle\sum_{j=1}^n\eta_{hi}\eta_{tj} \operatorname{Cov}( x_{i}, x_{j}) \\
= 
  \tfrac{1}{(n-1)^2}\biggl\{\displaystyle\sum_{i=1}^n \eta_{hi}\eta_{ti} {(\ell-1)(n- \ell)}
  - 2 \displaystyle\sum_{i < j} \eta_{hi}\eta_{tj} \tfrac{(\ell-1)(n-    \ell)}{n-2}\biggr\}\\
 = 
 \tfrac{(\ell-1)(n- \ell)}{(n-1)^2}   \biggl( \displaystyle\sum_{i=1}^n \eta_{hi}\eta_{ti} -
  {2}\displaystyle\sum_{i < j} \tfrac{\eta_{hi}\eta_{tj}}{n-2} \biggr)
\end{multline*}
Note that  $E(k_{th}) {=} n k(k{-}1)/\{(n{-}1)(n{-}2)\} {=} n\theta (k{-}1)/(n{-}2).$
Therefore
\begin{equation}
 \operatorname{Cov}( r_{\ell,t}, r_{\ell,h} ) 
=-\tfrac{(\ell-1)(n- \ell)(n-2k)}{(n-1)(n-2)^2} \theta.
\label{eq:cov}
\end{equation}
These covariances tend to zero uniformly on $\ell {=} o(n)$ as  $n{\rightarrow}+\infty$.
For the variance, we have
\begin{align*}
  \operatorname{Var}(r_{\ell}  )
  = \sigma_{\ell}^2 \left\{ 1 - \tfrac{(\ell -1)(n-2k)}{(n-2)(n - k -1)} \right\},
\end{align*}
from where the proposition follows.

\end{proof}

%
%
Let $s_\bullet = \max( 0, k -n + \ell) \, \ell/2$, 
$s^\bullet = \min(\ell-1, k) \, \ell/2$ and $\sigma_{\rho,\ell} = \sigma_l\sqrt{1-\rho_\ell}$.
Set $w(\ell, s) = 2\ell_2 \{ (s + \tfrac{1}{2})/\ell_2 - \nu_{\ell}\} / \sigma_{\rho,\ell} $,
where  $\nu_{\ell} = \theta + \beta\sigma_{\rho,\ell}^2$.
Let $\Phi(\cdot)$ denote the standard normal cumulative distribution function (CDF), and set
$\Delta_\Phi(\ell) = \Phi( w(\ell,  s^\bullet + \tfrac{1}{2})) -    \Phi( w(\ell,  s_\bullet - \tfrac{1}{2})).$
The above observations lead us to the following fact that backgrounds the main result of this section.
\begin{fact}
The MGF $ {\cal M}_{p_d(\cdot|\ell)}(2\beta)$ is well-approximated by 
$\exp\{ 2\beta  \ell_2\theta  + \tfrac{\beta^2}{2} \sigma^2_{\rho, \ell} \} \Delta_\Phi(\ell)$.
Therefore the main bulk of the sum in \eqref{eq:partition:1}
is well approximated by 
\[ \displaystyle\sum_{\ell \gg 1}^{o(n)} \binom{n}{\ell} \exp(\alpha' \ell) \exp\{ 2 \beta  \ell_2\theta  
+  \beta^2 \ell_2 \theta (1 - y_{\ell,2}) (1 - \theta)(1  - \rho_{\ell}) \} \Delta_\Phi(\ell).
\]
\end{fact}
In fact, it is well known that the hypergeometric distribution is well approximated by a normal distribution provided that
  its variance goes to infinity (see p. 158, Section 4.4 of \citep{govindarajulu-1965}). In our case, this means that the result is valid
  when $\ell = o(n)$.
 Also Hoeffding's inequality (see Section 6 of \citet{Hoeffding-1963}) for
  hypergeometric variables, states that each $r_{\ell,h}$ is concentrated about its mean $\mu_{\ell,1}= (\ell-1)\theta$, when $n$ is large and $\ell$ is moderate 
  to large. So we just need to study the distribution about its mean.
%
  %
%
These results
imply  that the variables $r_{\ell,h}$ are well-approximated by, and behave like, normally distributed random variables with mean $\mu_{\ell,1}$, and variance $\sigma_{\ell,1}^2 \approx \sigma_{\ell}^2/\ell $,
for values of 
$\ell = o(n)$.
Moreover, from Proposition~\ref{prop:var}, these Gaussian variables are
nearly independent. Therefore their sum is also well-approximated by a normal random variable.
That is, 
for large $\ell = o(n)$, 
we have
\begin{equation}
  \begin{split}
{\cal   M}_{p_d(\cdot|\ell)}(2\beta) \approx\int_{2s_\bullet-1}^{2s^\bullet+1}
\tfrac{1}{\sigma_{\rho,\ell}} \phi\biggl( \tfrac{x - \mu_{\ell}}{\sigma_{\rho,\ell}} \biggr)\exp( \beta x )
 dx,
\end{split}
\label{eq:eq4}
\end{equation}
where $\phi(\cdot)$ stands for the standard normal density.
A straightforward calculation simplifies \eqref{eq:eq4} to
\begin{equation*}
\exp\left(2\beta \theta \ell_2 + \tfrac{\beta^2}2 \sigma_{\rho,\ell}^2\right) \Delta_{\Phi}(\ell),
\end{equation*}
which is the desired result.
\vspace{10pt}

We propose an estimator of the partition function using the above results. 
To avoid the approximation at the extreme values of $\ell$,
set $A_{\phi}(\alpha, \beta) =  1 + \exp( \alpha' n + \beta k n)
+n \exp(\alpha')$. Also, let $$g(\ell) = \alpha' \ell + 2\beta \theta
\ell_2 \{ 1 + (\beta/2)\left(1 {-} \tfrac{\ell-2}{n-2} 
\right)(1{-}\theta)(1 {-}\rho_{\ell}) \}$$ and
define $\Sigma(\alpha, \beta) {=} \displaystyle\sum_{\ell= 2}^{n-1} \binom{n}{\ell}
  \exp\{ g(\ell) \} \Delta_\Phi(\ell).$
  Our partition function estimate
$Z_{\phi}(\alpha, \beta)$,  which we refer to as the
{\em normal-edge partition function estimate} 
is 
\begin{equation}
\begin{split}
Z_{\phi}(\alpha, \beta)=& A_{\phi}(\alpha, \beta) + \Sigma(\alpha, \beta).
\end{split}
  \label{eq:Z:phi}
\end{equation}
The estimate  
$Z_{\phi}(\alpha, \beta)$  may be computed in ${\cal O}(n)$
operations. Although, this calculation is fast, we investigate a
further approximation by replacing the summation in the last term of 
\eqref{eq:Z:phi} by an integral. Specifically, we regard the summation
as a Riemann sum, and hence, as an approximation to the corresponding
integral. The integral can be calculated as a new 
summation with number of terms much smaller than $n$, yielding a final
approximation that can be computed much faster than ${\cal O}(n)$. 
In fact, the Euler-MacLaurin formula
\citep{abramowitzandstegun64} gives us  the  approximation 
\begin{equation*}
 \Sigma(\alpha, \beta)\approx  \int_{2}^{n-1} \tfrac{\Gamma(n+1)}{\Gamma( x) \Gamma(n-x)} \exp\{ g(x)\}
  \Delta_\Phi(x)\, dx + 
  B_{\phi}(\alpha,\beta),
  \end{equation*}
  where
  $2B_{\phi}(\alpha, \beta) =  n \exp\{g(n-1)\} \Delta_\Phi(n-1)  + {n\choose2} \exp\{ g(2) \} \Delta_\Phi(2).$
Using \citet{stirling1730}'s approximation for the Gamma function, and
writing $y= x/n =  \ell/n$,
in
the above integral, we have
$  \Sigma(\alpha, \beta)
  \approx 
  B_{\phi}(\alpha, \beta) + J_{n,k,\frac12}(\alpha',\, \beta),
$
  where, for every $t\in \mathbb{R}$,  we have the function
  \begin{equation}
    J_{n,k,t}(\alpha,\, \beta)  =  \sqrt{\tfrac{n}{2\pi}}
    \int_{2/n}^{1 - 1/n}
    \tfrac{\Delta_{\Phi}(ny) \exp\{ g(ny)\} }{
    (1-y)^{n(1-y) + \frac{1}{2}} y^{ny+ t }}
    dy.\label{eq:J}
  \end{equation}
We  denote this integral approximation to $Z_{\phi}(\alpha, \beta)$
by $\tilde{Z}_{\phi}(\alpha, \beta)$.

This integral approximation error is of order  $\max_y |e(y)|/n$ where
$e(y)$ is the derivative of the integrand in \eqref{eq:J} and is equal
to $\exp\{h(y) + g(y)\} ( \Delta_\Phi(ny) (h'(y) + g'(y)) +
\Delta_\Phi'(ny))$ with $h(y) =  -[n(1 -y)+\tfrac{1}{2}]\log( 1-y) -(ny +\tfrac{1}{2})\log(y)$.
This means that the relative error of the integral approximation is of
order $n^{-1}$ because the terms in the summation are at most of the
same order as $\max_y|e(y)|$.
\subsubsection{Moment approximations}
Having found approximations for $Z(\alpha,\beta)$, we now approximate
the first two moments of \eqref{eq:Ising} under isotropy and homogeneity.
\paragraph{Mean number of active nodes}
\label{approx:M}
We first approximate $M {=} \E(\hat M) {\equiv} \E( \sum_{i=1}^n x_i)$, the expected number of
active nodes. 
  In addition to the setup in Section~\ref{approx:Z}, let
$\Delta_\phi(\ell) = \phi\{ w(\ell,  s^\bullet {+} \tfrac{1}{2}) \} {-} \phi\{  w(\ell,  s_\bullet {-} \tfrac{1}{2}) \}$.
Define
$  C_{\phi.M}(\alpha, \beta)= 
  {n\choose2}  [ \exp\{ g(n{-}1)\}  \Delta_\Phi(n{-}1) {+} \exp\{ g(2)\} \Delta_\Phi(2) ].$
Using the same reasoning as before yields the following estimate for $M$:
\begin{equation}
  \begin{split}
M_\phi \doteq \tfrac{1}{{Z}_\phi} 
\biggl\{
n \exp( \alpha' n + \beta k n)
+n \exp(\alpha')   +  \displaystyle\sum_{\ell=2}^{n-1} \binom{n}{\ell} \ell \exp\{ g(\ell)\}
   \Delta_\Phi(\ell) \biggr\}.
   \label{eq:Mphi}
\end{split}
\end{equation}
Using \citet{stirling1730}'s approximation, the Euler-MacLaurin
expansion~\citep{abramowitzandstegun64}, and writing $\ell$ as $n
(\ell/n)$, provides an approximation for the series in \eqref{eq:Mphi} as
$C_{\phi.M}(\alpha, \beta) {+} n J_{n,k,-\frac12}(\alpha,\, \beta)$, and thence
\begin{equation*}
  \begin{split}
  \tilde M_\phi =  \tfrac{n}{\tilde Z_\phi}  \bigl\{
  \exp{ (\alpha'n  + \beta k n )} + \exp(\alpha') + J_{n,k,
    -\frac12}(\alpha,\beta)   +  C_{\phi.M}(\alpha, \beta) / n \bigr\}.
\end{split}
\end{equation*}

\paragraph{Mean spin interaction}
\label{approx:S}
We now approximate the  expected number
of matches of the homogeneous isotropic Ising model, or its mean spin interaction $S= \E(\hat S)=\E(\tfrac{1}{2}\sum_{i,j} \eta_{ij} x_i x_j
)$. Writing $Z{=}Z(\alpha, \beta)$, and proceeding as before,
\begin{equation*}
  \begin{split}
    S 
    & =
    \tfrac{1}{Z} 
    \displaystyle\sum_{\ell=0}^n \exp({\alpha' \ell}) 
    \displaystyle\sum_{\bx \in M(\ell)} \!\!\bigl(\tfrac{1}{2} \displaystyle\sum_{i, j} \eta_{ij} x_i x_j \bigr) \exp( \beta \displaystyle\sum_{i, j} \eta_{ij} x_i x_j )\\
  &  = \tfrac{m}{Z}  \exp({\alpha' n + \beta k n }) +
    \tfrac{1}{Z} \displaystyle\sum_{\ell=2}^{n-1} \binom{n}{\ell} \exp({\alpha' \ell})
    \displaystyle\sum_{s=s_\bullet}^{s^\bullet} p(s|\ell) \, s \exp(2\beta s).
  \end{split}
\end{equation*}
From weak convergence arguments and similar reductions as in the
approximation for the partition function
$Z(\alpha,\beta)$, 
we get
\begin{equation*}
\begin{split}
\E\{\hat S\exp{(2 \beta \hat S)}\}&\approx \int_{2s_\bullet-1}^{2s^\bullet+1}
\tfrac{1}{\sigma_{\rho,\ell}} \phi\biggl( \tfrac{x - \mu_{\ell}}{\sigma_{\rho,\ell}} \biggr)
\tfrac{x}{2} \exp( \beta x ) \, dx\\
&= \tfrac{1}{2} \tfrac{d}{d\beta}
\int_{2s_\bullet-1}^{2s^\bullet+1}
\tfrac{1}{\sigma_{\rho,\ell}} \phi\biggl( \tfrac{x - \mu_{\ell}}{\sigma_{\rho,\ell}} \biggr)
\exp( \beta x ) \, dx\\
  &=  \tfrac{1}{2}\exp\biggl(2\beta \theta \ell_2 + \tfrac{\beta^2 \sigma_{\rho,\ell}^2}2\biggr) \left\{ (2\theta \ell_2 + \beta\sigma_{\rho,\ell}^2 
    )\Delta_{\Phi}(\ell) - \sigma_{\rho,\ell} \Delta_{\phi}(\ell) \right\}.
\end{split}
\end{equation*}
 So we approximate $\E(\hat S)$ as
 \begin{multline}
\label{eq:s:phi}
S_\phi\doteq \tfrac{m}{Z_\phi}   \exp(\alpha'n + \beta k n)  
 + \tfrac{1}{2Z_\phi} \displaystyle\sum_{\ell=2}^{n-1} \binom{n}{\ell} \biggl[
  \{ 2\theta \ell_2 + 
  \beta\sigma_{\rho,\ell}^2  \} \Delta_{\Phi}(\ell) 
-  \sigma_{\rho,\ell}  \Delta_{\phi}(\ell) \biggr]\\
\qquad\times \exp\biggl(\alpha'\ell + 2\beta \theta \ell_2 + \tfrac{\beta^2 \sigma_{\rho,\ell}^2 }2\biggr).
\end{multline}
Let $\zeta_1(\ell) {=}  1 {+} \tfrac{\beta\sigma_{\rho,\ell}^2}{2\ell_2\theta}$,
and $\zeta_2(\ell) {=} \sigma_{\rho,\ell}   \Delta_{\phi}(\ell)$.
Define
$ D_{\phi.M}(\alpha, \beta) \doteq \tfrac{1}{2}\{ \zeta_1(2) \Delta_{\Phi}(2) - \tfrac{n{-}1}{2k}\zeta_2( 2)\}
  \exp\{ g(2)\}  +
  \{ \tfrac{n-2}2 \zeta_1(n-1) \Delta_{\Phi}(n{-}1)  {-}\tfrac{ \zeta_2( n{-}1)}{2k}\}
  \exp\{ g(n{-}1)\}.
  $
  The expression in (\ref{eq:s:phi}) can be further approximated as an integral in   the same manner as before to obtain the approximation
\begin{multline*}
 \tilde S_\phi \doteq
\tfrac{m}{\tilde Z_\phi} \biggl[
\exp{(\alpha' n + \beta kn )} + D_{\phi.M}(\alpha, \beta) \\
 \qquad+\sqrt{\tfrac{n}{2\pi}} \int_{\tfrac{2}{n}}^{1 - \tfrac{1}{n}} 
\biggl\{ \zeta_1(ny)\Delta_{\Phi}(ny) -
 \tfrac{\zeta_2(ny )}{kny^2}\biggr\} 
\tfrac{\exp[\alpha' ny + \tfrac{\beta k n y^2}2 \zeta_1(ny)  ] }
     {(1-y)^{n(1-y) + \frac{1}{2}} y^{ny - \frac{3}{2}} } \, dy \biggr] ,
\end{multline*}
upon replacing $kn/2$ by $m$, using \citet{stirling1730}'s approximation 
to the binomial coefficients, and the Euler-MacLaurin approximation~\citep{abramowitzandstegun64} for the sum.

\subsection{Extension to the anisotropic case}
\label{sec:anisotropic}
We now explore the case of the homogeneous anisotropic Ising model. 
We consider two sets of edges $E_p, E_v$ induced by the neighborhood relations
$E_p\! =\! \{ (i,j) : i\! \simp\! j\}$, $E_v\! =\! \{ (i,j) : i\!
\simv\! j\}$, with 
$E_p \cap E_v = \emptyset$. For instance, we may think of $E_p$ as the
set of edges in the same ``plane'' (or a 2D lattice)
and of $E_v$ as the set of edges formed by two ``planes'' (of the
third dimension in a 3D-lattice). Then the Ising model PMF is 
\begin{multline}
\mP(\bx;\alpha,\beta_p, \beta_v) = Z(\alpha,\beta_p, \beta_v)^{-1} \exp\left\{
\alpha'' \displaystyle\sum_{i=1}^n x_i 
+ \beta_p \displaystyle\sum_{i j} \eta_{p,ij} x_ix_j + \beta_v  \displaystyle\sum_{i j} \eta_{v,ij}x_ix_j\right \},
\label{eq:pmf.aniso}
\end{multline}
with $\alpha'' = \alpha -\beta_p k_p -\beta_v k_v,
\eta_{p,ij} = \I(i \simp j), \eta_{v,ij}=\I(i \simv j)$, 
and where $k_p$ and $k_v$ are the degrees of the regular graphs $(V,
E_p)$ and $(V, E_v)$. Proceeding  similarly as before, 
$  Z(\alpha,\beta_p, \beta_v)
 =  1 + \exp{(\alpha n)} + n\exp(\alpha'')
+\displaystyle\sum_{\ell=2}^{n-1} \! \! \exp(\alpha'' \ell) \!\!\!
\displaystyle\sum_{X \in M(\ell)} \!\!\! \exp\bigl(\beta_p \displaystyle\sum_{i, j} \eta_{p, ij} x_i x_j + \beta_v \displaystyle\sum_{i, j} \eta_{v, ij} x_i x_j\bigr)$. 
From identical arguments as in Proposition~\ref{prop:Z}, 
we have
\begin{equation*}
Z(\alpha,\beta_p, \beta_v)
= A_2(\alpha,\beta_p,\beta_v) 
+ \displaystyle\sum_{\ell=2}^{n-2} {n\choose\ell} \exp{(\alpha'' \ell)} {\cal M}_{p(\cdot,\cdot | \ell)}( 2\beta_p, 2\beta_v),
\end{equation*}
	where $A_2(\alpha,\beta_p,\beta_v) =1  + \exp{(\alpha n)} + n\exp{(\alpha'')}  +
n \exp{\{ \alpha (n-2) + \alpha'' \}},$
and 
\[
  {\cal M}_{p(\cdot,\cdot | \ell)}( 2\beta_p, 2\beta_v)  = \!
\displaystyle\sum_{\mathclap{\substack{s_p, s_v=0\\  s_p+s_v \leq \ell_2}}}  \exp\{ 2\beta_p s_p + 2\beta_v s_v \}
\tfrac{Y(s_p, s_v, \ell)}{ {n\choose\ell} },
\]
  with $Y(s_p, s_v, \ell)$ equal to the number of node-induced graphs
  $G_\ell$ 
  containing exactly $s_p$ edges in $E_p$ and exactly $s_v$ edges in $E_v$.
  As before, consider $r_{p,\ell, h} = \sum_{j=1}^n \eta_{p, i_h, j} x_j$, and similarly,
  $r_{v,\ell, h} = \sum_{j=1}^n \eta_{v, i_h, j} x_j$, where $\{x_{i_1}, \ldots, x_{i_\ell} \}$ is the set
  of nodes of the graph 
    $G_\ell$.
  Let $r_{p,\ell} {=} \sum_{h=1}^\ell r_{p,\ell, h}$, and $r_{v,\ell} {=} \sum_{h=1}^\ell r_{v,\ell, h}$.
  As in the derivation of $Z_{\phi}$ earlier,
  we argue that the joint PMF $p(r_{p,\ell},r_{v,\ell} \lvert \ell)$ behaves like a
  bivariate normal distribution, and hence ${\cal M}_{p(\cdot,\cdot | \ell)}( 2\beta_p, 2\beta_v)$
  is approximately
  \begin{equation} \label{eq:mgf:2}
    \int_{2s_{p\bullet} - 1}^{2s_p^\bullet + 1} \int_{2s_{v\bullet}-1}^{\min(2s_v^\bullet + 1, 2s_{pv}^\bullet + 1 - x)}
   \hspace{-0.1in} \tfrac{ \exp( \beta_p x + \beta_v y) }
         {2\pi \sqrt{\operatorname{det}(V_\ell)}} \exp\{ {-} (x {-}\mu_{p,\ell}, y {-}\mu_{v, \ell}) V_\ell^{-1} (x {-}\mu_{p,\ell},y{-}\mu_{v,\ell})' \}
         \, dx dy,
    \end{equation}
  where $s_{ a \bullet} = \max(0, k_a -n + \ell) \ell/2$, $s_a^\bullet = \min(\ell-1,k_a) \ell/2$,
  for $a =p,v$,
  $s_{pv}^\bullet = \min( \ell_2, (k_p+k_v)\ell/2 )$,
 $\mu_{p,\ell} = 2\ell_2 \theta_p$, and $\mu_{v, \ell} = 2 \ell_2 \theta_v$, with
  $\theta_p = k_p / (n-1),$ and $\theta_v = k_v / (n-1)$;
and 
  \[
  V_\ell = \begin{pmatrix}
    \sigma_{p,\ell}^2 (1 - \rho_{p,\ell}) &   \sigma_{p,v,\ell} \\
    \sigma_{p,v,\ell}  &  \sigma_{v,\ell}^2 (1 - \rho_{v,\ell})
  \end{pmatrix}.
  \]
  Let
  $\tau_{p,\ell}^2 = \sigma_{p,\ell}^2 (1 -\rho_{p,\ell})$, and
  $\tau_{v,\ell}^2 = \sigma_{v,\ell}^2 (1 -\rho_{v,\ell})$.
  A straightforward calculation shows that \eqref{eq:mgf:2} reduces to
  \[ \exp\biggl\{ \beta_p \mu_{p,\ell} + \beta_v \mu_{v,\ell}
  + \tfrac{1}{2} (\beta_p, \beta_v) V_\ell (\beta_p, \beta_v)' \biggr\}
  \Delta_{\Phi_2}(\ell),
  \]
  where
  $\Phi_2(\cdot)$ is the standard bivariate normal CDF, and
  $\Delta_{\Phi_2}( \ell)$ is $\Phi_2(\cdot)$ evaluated in the region
  $V_\ell^{-1/2}\bigl[ \Omega_1 - \{(\mu_{p,\ell}, \mu_{v,\ell})' +  V_\ell (\beta_p, \beta_v)'\} \bigr]$,
  with $\Omega_1 = [2s_{p\bullet} - 1, 2s_p^\bullet + 1]\times [2s_{v\bullet}-1, 2s_v^\bullet + 1]$.
From here, approximation formulae analogous to the ones in 
\eqref{eq:Z:phi} and \eqref{eq:J} are readily obtained for the anisotropic case.
For example, let $g_2(\ell) = \alpha'' \ell + \beta_p \mu_{p,\ell} + \beta_v \mu_{v,\ell}
+ \tfrac{1}{2} (\beta_p, \beta_v) V_\ell (\beta_p, \beta_v)'$, and
$2 B_{2, \phi}(\alpha, \beta_p, \beta_v) =  {n \choose 2}\exp\{g_2(n-2)\}\Delta_{\Phi_2}(n-2) + {n \choose 2}\exp\{g_2(2)\}\Delta_{\Phi_2}(2)$.
We have the approximation
\begin{equation}
  \begin{split}
  \tilde{Z}_{\phi}(\alpha, \beta_p, \beta_v) =
&  A_2(\alpha,\beta_p,\beta_v) 
  + B_{2, \phi}(\alpha, \beta_p, \beta_v)  \sqrt{\tfrac{n}{2\pi}}    \int_{2/n}^{1 - 2/n}
    \tfrac{\Delta_{\Phi_2}(ny) \exp\{ g_2(ny)\} }{
    (1-y)^{n(1-y) + \frac{1}{2}} y^{ny+ \frac{1}{2} }}
      dy.
      \end{split}
      \label{eq:Z2}
\end{equation}
\begin{remark}
The variance-covariance $V_\ell$ is obtained in a similar manner as the variances and covariances in the single graph case.
We have
\begin{multline*}
 \operatorname{Cov}( r_{p,\ell,t}, r_{v, \ell,h}\lvert \eta )  = 
 \displaystyle\sum_{i=1}^n
\displaystyle\sum_{j=1}^n\eta_{v,hi}\eta_{p,tj} \operatorname{Cov}( x_{i}, x_{j}) \\
= 
  \tfrac{(\ell-1)(n-\ell)}{(n-1)^2} \left\{ \displaystyle\sum_{i=1}^n \eta_{p,ti}\eta_{v,hi}
  - \tfrac{2}{n-2} \displaystyle\sum_{i < j} \eta_{v,hi}\eta_{p,tj} \right\}= 2 k_p k_v  \tfrac{ (\ell-1)(n- \ell) } { (n-1)^2(n-2)^2 }.
\end{multline*}
\end{remark}

 Next we 
proceed as in Section~\ref{sec:isotropic} to obtain
$M_\phi$, $\tilde{M}_\phi$, $S_{\phi}$ and $\tilde{S}_{\phi}$. However, now we need to estimate
both $S_{p, \phi}$ and $S_{v,\phi}$. The estimates 
  $M_{\phi}, S_{p, \phi},  S_{v, \phi}$ correspond to the derivatives of $\log Z_{\phi}(\alpha, \beta_p, \beta_v)$ with respect to $\alpha$, $\beta_p$,
  and $\beta_v$, respectively. The approximations based on integrals instead of summations can be seen
  to correspond to derivatives of $\log Z_{\tilde \phi}(\alpha, \beta_p, \beta_v)$ as well.
  Therefore, estimates of $\tilde M_\phi$, $\tilde{S}_{p, \phi}$ and $\tilde{S}_{v, \phi}$ can be obtained directly from derivatives
  of $Z_{\tilde \phi}(\alpha, \beta_p, \beta_v)$.
    For brevity, we sketch the ideas here, and refer to Section~\ref{supp:approximations} for explicit approximations to
    $\E(M)$, $\E(S_p)$     and $\E(S_v)$.


\section{Performance evaluations}
\label{simulations}
\subsection{Approximation formulae assessments}
\label{approximation.assessments}

We evaluated performance of the analytical approximation  formulae 
derived in Section~\ref{approximations} by comparing them with those
obtained by simulation. The mean activation and spin interaction were estimated 
for a range of $(\alpha, \beta)$-pairs using  MCMC  -- these
estimates were assumed to be the ``gold standard'' for our
comparisons. However,  MCMC simulation-based estimation of 
 $Z(\alpha,\beta)$ is very difficult, so we used path
sampling~\citep{Ogata-1989,Gelman-Meng-1998}  to 
obtain its reference value.
Our approximation formulae apply to any regular graph,
but for convenience, we only evaluated performance on lattice
 graphs. (Because of  edge effects, our lattice
 graphs are only approximately regular.) 
 We note also that the use of MCMC methods here is simply to assess the accuracy of our formulae: while faster stochastic approaches may be resorted to, we consider it largely irrelevant in our comparisons, simply because it is inherently unfair to the stochastic methods to compare the performance of our approximation formulae that are completely scalable and near-immediately calculated, to any  stochastic method, that generated the Ising model and is affected by its size. Therefore, the real evaluation of our method is in its accuracy. Also, because our approximation formulae rely on asymptotic reductions, it is of special interest to understand their accuracy in less favorable contexts where asymptotic arguments may be more suspect. Therefore, our evaluations are on moderate, but realistic-sized fields. We also reiterate that formulae rely on asymptotic reductions, and therefore it is of special interest to understand their accuracy in less favorable contexts where asymptotic arguments may be more suspect. 
 We consider the isotropic  and anistropic cases separately.
 \subsubsection{The isotropic case}
 \label{isotropy}
We simulated 
realizations from Ising models on two lattice configurations and with
three different neighborhood orders.  
The two lattices had grids of sizes $116{\times} 152$ and $64{\times} 64$, respectively. The neighborhoods we chose for our simulations were of the first,
second and fifth orders, corresponding to graphs of degree $k{=}4, 8$
and $24$, respectively. For each of the six  combinations of grid
sizes and graph degrees, we compared performance for {1,102} different
pairs of values of the 
 Ising parameters $(\alpha, \beta) {\in} [0,\, 5] {\times} [0.005, \, 10]$
(19 values for $\alpha$, and 58 values for $\beta$). 
Note that there is no need to evaluate the approximations for negative
$\alpha$ because
$Z(-\alpha, \beta ) {=} \exp{(-\alpha n)} Z(\alpha, \beta)$ for all
pairs $(\alpha, \beta)$. (In particular, moments such as 
$\E_{-\alpha, \beta}(\hat{M})$
can be easily obtained from $\E_{\alpha, \beta}(\hat{M})$.)
For each
setting, we estimated the Ising moments and normalizing constant from
samples obtained using the Swendsen-Wang~\citet{Swendsen-Wang-1987} algorithm
with a burn-in period of 10,000 iterations and a sample size of 
10,000 realizations from the post-burn-in iterations and used these
estimates as the ``gold standard'' reference values. 
For each moment estimate, we evaluated {the} performance of our
approximations relative to the  MCMC estimate by computing both 
the absolute value difference between the
MCMC estimate ${m_{MC}}$ and the analytical approximation given by the
Normal edge proportion approximation ${m_{N}}$, 
and the relative absolute difference between these quantities
$|m_{MC} {-} m_{N}| / m_{MC}.$ 
{ The measures of absolute ($L_1$) and
relative ($R_1$) discrepancy between all evaluations in the grid for
$(\alpha, \beta)$ 
are given by the difference between the approximated and estimated
surfaces
$$ L_1( m_{MC}, m_{N}) {=} V^{-1} \int\int |m_{MC}(\alpha,\beta) - m_{N}(\alpha,\beta)|
\, d\alpha\, d\beta,$$
and 
$$R_1( m_{MC}, m_{N}) {=} V^{-1}\int\int \frac{|m_{MC}(\alpha,\beta) - m_{N}(\alpha,\beta)|}{|m_{MC}(\alpha,\beta)|} \, d\alpha\, d\beta,$$
with $V{=} \int\int  d\alpha\, d\beta.$
We also show the ratio of the absolute discrepancy to the mean volume of the
region below the surface, given by $m_{MC}$,
$L_1(  m_{MC}, m_{N})/ V_{MC}$,
where $V_{MC}{=} (V^{-1} \int\int  m_{MC} d\alpha\, d\beta)$.}
Table~\ref{table:logZ} shows the relative { and absolute discrepancies}
between the 
analytical approximations of $\log Z(\alpha, \beta)$ and the path sampling estimates using the MCMC samples.
The path sampling estimates were obtained using the estimate of the expected
matches, that is 
$\log Z_{MC}(\alpha, \beta) {=}
\int_{0}^{\beta} \operatorname{Mean}_{\alpha, b}(\sum_{i\sim j} \delta_{ij}^{(t)}) 
\, db 
+ n\log\{ 1 + \exp{(\alpha)}\} {-} m \beta,$
where $\delta_{ij}^{(t)}$ is the observed value of $\delta_{ij}$ in the
$t$th sample generated by the Swendsen-Wang algorithm and, as before,
$m$ is the number of edges in the graph. 
\begin{table}
\centering
  \caption{Discrepancies associated with the approximation of the logarithm of
  the partition function\label{table:logZ}. Here, we have two lattice
  grids ($\G$), A of size $116\!\times\! 152$ and $B$ of size
  $64\times64$ respectively, 
  with regular graph of degree $k$.
 All discrepancies are computed against the MCMC
 path sampling estimates which forms our ``gold standard'' for comparisons in these experiments,
 except for the case $k{=}2$, whose discrepancies were computed using the known asymptotic formula for the 2-NN graph. 
}
\begin{tabular}{rrrrllll} \\ \hline
  \multicolumn{2}{c}{\ } & \multicolumn{6}{c}{Absolute and Relative Discrepancies} \\
  \multicolumn{2}{c}{\ }
& \multicolumn{2}{c}{$L_1$} &  \multicolumn{2}{c}{$L_1 / V_{MC} $} &
  \multicolumn{2}{c}{$R_1$} \\ \cline{3-8}
  \multicolumn{8}{c}{\  }\\
\multicolumn{1}{c}{$\G$} & \multicolumn{1}{c}{$k$} &
\multicolumn{1}{r}{$\tilde Z_{\phi}$} & \multicolumn{1}{r}{$Z_{\phi}$} &
\multicolumn{1}{r}{$\tilde Z_{\phi}$} & \multicolumn{1}{r}{$Z_{\phi}$}  &
\multicolumn{1}{r}{$\tilde Z_{\phi}$} & \multicolumn{1}{r}{$Z_{\phi}$} \\\hline
A & $2$  &  $26.72$  & $26.73$  & $0.0006$& $0.0006$& $0.009$ & $0.009$ \\
A &  $4$ &  $270.48$ & $270.55$ & $0.006$ & $0.006$ & $0.032$ & $0.032$ \\
A &  $8$ &  $437.97$ & $438.00$ & $0.010$ & $0.010$ & $0.047$ & $0.047$ \\
A & $24$ &  $373.57$ & $373.59$ & $0.008$ & $0.008$ & $0.044$ & $0.044$ \\
B   & $2$  &  $6.20$   & $6.20$   & $0.0006$& $0.0006$& $0.009$ & $0.009$ \\
B   &  $4$ &  $62.86$  & $62.93$  & $0.006$ & $0.006$ & $0.032$ & $0.032$ \\
B   &  $8$ & $100.81$  &$100.81$  & $0.010$ & $0.010$ & $0.047$ & $0.047$ \\
B   & $24$ &  $88.24$  & $88.24$  & $0.009$ & $0.009$ & $0.044$ & $0.044$ \\\hline
\end{tabular}
\end{table}
Table~\ref{table:logZ} indicates that both the direct Normal edge
proportion estimate and its counterpart 
that uses the Euler-MacLaurin formula perform similarly. Thus there is almost
no loss in accuracy when using the faster Euler-MacLaurin based estimate.
In order to further show the value of this simplified approximation, we  
also compared  this 
analytical approximation
with the theoretical asymptotic result for the logarithm of the normalizing constant for the 2-nearest-neighbor (2-NN) graph, that is, for $k{=}2$.
Section~\ref{sec:appendix:1NN} shows that $\log Z(\alpha, \beta)$
approximately equals 
\begin{equation}
  \label{1nn}
   \tfrac {n (\alpha - \beta ) }2
+ n\log \biggl[ \exp(\beta/2 ) \cosh(\alpha/2) 
+ \sqrt{ \exp(\beta ) \cosh^2(\alpha/2) - 2\sinh(\beta) } \biggr],
\end{equation}
after adapting the result of 
~\citet[Chapter 13, page
261]{Salinas-2001}
to the case of the \{0,1\}-statespace 2-NN Ising model.
Performance evaluations for this case are also in Table~\ref{table:logZ},
with the results again indicating that the numerical approximation
works very well even for 
the smallest possible value of $k$  even though these
approximations are based on moderate to large values of $k$.
\begin{table}[h]
\begin{center}
\caption{Absolute and relative discrepancies between MCMC estimates
of (a) the mean number of active nodes or (b) the mean spin interaction, and
their numerical approximations for
lattice grids ($\G$) A of size $116\!\times\!152$ and B of size
$64\!\times \!64$, and degree $k$. 
}
  \label{table:sim:comparison:L1}
\subfloat[Mean number of active nodes]{
\begin{tabular}{rrrrllll} \\ \hline
\multicolumn{2}{c}{\  } & \multicolumn{6}{c}{Absolute and Relative Discrepancies} \\ 
\multicolumn{2}{c}{\  } & \multicolumn{2}{c}{$L_1$} &  \multicolumn{2}{c}{$L_1 / V_{MC} $} & \multicolumn{2}{c}{$R_1$} \\ \cline{3-8}
\multicolumn{8}{c}{\  }\\
\multicolumn{1}{c}{$\G$} & \multicolumn{1}{c}{$k$} &
\multicolumn{1}{c}{$\tilde M_{\phi}$} & \multicolumn{1}{c}{$M_{\phi}$} &
\multicolumn{1}{c}{$\tilde M_{\phi}$} & \multicolumn{1}{c}{$M_{\phi}$}  &
\multicolumn{1}{c}{$\tilde M_{\phi}$} & \multicolumn{1}{c}{$M_{\phi}$} \\\hline
A  &  $4$ & $18.94$ & $19.71$ & $0.001 $  & $0.001 $  & $0.002$  & $0.002$ \\
A  &  $8$ & $8.42$  & $ 8.81$ & $0.0005 $ & $0.0005$  & $0.0009$ & $0.0009$  \\
A  & $24$ & $2.84$  & $ 2.94$ & $0.0002 $ & $0.0002$  & $0.0003$ & $0.0003$ \\
B  &  $4$ & $ 4.14$ & $ 4.32$ & $0.001 $  & $0.001 $  & $0.002$  & $0.002$ \\
B  &  $8$ & $ 0.73$ & $ 0.74$ & $0.0002 $ & $0.0002$  & $0.0003$ & $0.0003$ \\
B  & $24$ & $ 0.25$ & $ 0.25$ & $0.00006$ & $0.00006$ & $0.0001$ &
                                                                   $0.0001$ \\ \hline
\end{tabular}
}

  \subfloat[Mean spin interaction]{  
    \begin{tabular}{rrrrllll} \\
      \hline
\multicolumn{2}{c}{\  } & \multicolumn{6}{c}{Absolute and Relative Discrepancies} \\ 
\multicolumn{2}{c}{\  } & \multicolumn{2}{c}{$L_1$} &  \multicolumn{2}{c}{$L_1 / V_{MC} $} & \multicolumn{2}{c}{$R_1$} \\ \cline{3-8}
\multicolumn{8}{c}{\  }\\
\multicolumn{1}{c}{$\G$} & \multicolumn{1}{c}{$k$} &
\multicolumn{1}{c}{$\tilde S_{\phi}$} & \multicolumn{1}{c}{$S_{\phi}$} &
\multicolumn{1}{c}{$\tilde S_{\phi}$} & \multicolumn{1}{c}{$S_{\phi}$}  &
\multicolumn{1}{c}{$\tilde S_{\phi}$} & \multicolumn{1}{c}{$S_{\phi}$} \\\hline
  A  &  $4$ & $39.18$ & $41.47$ & $0.001$   &  $0.001$  & $0.002$  & $0.002 $ \\
  A  &  $8$ & $33.57$ & $35.90$ & $0.0005$  & $0.0005$  & $0.0009$ & $0.001 $   \\
  A  & $24$ & $33.80$ & $34.65$ & $0.0002$  &  $0.0002$ & $0.0003$ & $0.0003$  \\
  B  &  $4$ & $9.27$  & $9.35$  & $0.001$   &  $0.001$  & $0.002$  & $0.002$ \\
  B  &  $8$ & $3.94$  & $3.00$  & $0.0003$  &  $0.0002$ & $0.0004$ & $0.0003$\\
  B  & $24$ & $3.85$  & $2.66$  & $0.00008$ &  $0.00006$& $0.0001$ & $0.0001$ \\
\hline\end{tabular}
}
\end{center}
\end{table}

 Table~\ref{table:sim:comparison:L1} reports the values of the
 absolute value and relative discrepancies for the relevant moments
 (mean activation and spin interaction) of the Ising model.
 The results indicate good performance of our
   approximation formulae relative to the MCMC estimates.
   It is worth noting that the MCMC algorithm
   took about four days to compute the 1,102 sets of moments for each
   combination of grid-size and graph degree combination, while our
   approximation formula took well under a second for each calculation. (We stress that our approximations result in formulae, so that our computational cost is unaffected by larger-sized MRFs).  
\subsubsection{The anisotropic model case}   
\label{anisotropy}
We also evaluated performance of the approximation formulae derived in
Section~\ref{sec:anisotropic}. 
\begin{table}
  \begin{center}
  \caption{Mean relative discrepancies between MCMC estimates and numerical approximations for the anisotropic case
    (standard deviations are shown within parentheses).
    The simulations are made on a lattice of dimension $66{\times} 106$
    to match the setting of the \texttt{pistachios} datasets.
}   \label{table:sim:aniso}
  \begin{tabular}{crrrrr} \\\hline
    \multicolumn{6}{c}{ Relative Discrepancies for the Anisotropic Graph} \\ 
    Graph & \multicolumn{2}{c}{degrees} & \multicolumn{3}{c}{Moments} \\ 
    Order & $k_p$ & $k_v$ & $ \tilde M_\phi$ & $\tilde S_p$ & $\tilde S_v$ \\ \hline
        \multicolumn{6}{c}{\  }\\
    1           & 2     & 2     &  0.060\,(0.043)  & 0.108\,(0.071) & 0.107\,(0.070) \\
    3           & 4     & 8     & 0.013\,(0.023)   & 0.024\,(0.042) & 0.025\,(0.043) \\
    5           & 4     & 20    & 0.006\,(0.018)   & 0.011\,(0.033) & 0.011\,(0.035) \\ \hline
  \end{tabular}
  \end{center}
\end{table}
Our experiments were on a $66{\times} 106$ lattice inspired by the application in Section~\ref{sec:pistachios}.
Table~\ref{table:sim:aniso} summarizes the results for the anisotropic
model applied on such lattices.
The simulations shows the mean and standard deviations of the relative discrepancies evaluated on a
 thousand parameters
$(\alpha, \beta_p, \beta_v)$ in the hypercube $[0,1]^3$ chosen by a
Latin Hypercube Sampling (LHS) scheme \citep{mckay2000comparison}  which is
a quasi-random sampling method. This statistical method
covers the space of possible points more effectively than a uniform
sample. 
 The moments estimates were computed using approximate derivatives for $\log \tilde{Z}_\phi(\alpha, \beta_p, \beta_v)$.
 The derivatives were approximated using Chebyshev polynomials
 \citep{numercialrecipesc}.
 (We note that these results are presented  in a different
format in Table~\ref{table:sim:aniso} than in
Table~\ref{table:sim:comparison:L1}. This is because for the
isotropic model, we 
use a 2D uniform grid so that the integrals can be estimated in two
dimensions. For the anisotropic case however, the grid would have been
too big to perform the above operation, so we used the LHS
scheme. Consequently, this is not an uniform sample, and so  the mean is not necessarily
an estimate of the integral. For this reason, we show means and
standard deviations evaluated over the 1102  points of the LHS, in
Table~\ref{table:sim:aniso}.) In any case, clearly, the approximation error becomes smaller when the order of
 the graph increases. As in the isotropic case, our results closely
 match our theoretical developments in Section~\ref{approximations}.
\subsection{Parameter estimation performance}
\label{sec:estimation}

In this section we study the estimation problem for both the isotropic and the
anisotropic model. For this latter model, we consider three parameters $(\alpha, \beta_p, \beta_v)$
as in Section~\ref{sec:anisotropic}. For both models we consider up to order five neighborhoods
in the lattices (or graphs, in general).
For the second and third order anisotropic models, $\beta_p$ is the Ising parameter
associated with the first and second order neighborhoods, respectively. The parameter
$\beta_v$ is associated with the furthest neighbors. For the first order model,
each one of $\beta_p$ and $\beta_v$ is associated with one of the two axes of the lattice.
For the fourth and fifth order models, $\beta_p$ is associated with the second order neighborhoods,
and $\beta_v$ with the remaining (furthest) neighbors. This choice is a compromise to have similar
number of neighbors associated with each parameter.

For the isotropic  model we selected a hundred points
$\theta {=} (\alpha, \beta_p) {\in} (0,1)^2$, while for the anisotropic model, we selected a hundred points 
$\theta {=} (\alpha, \beta_p, \beta_v) {\in} (0,1)^3$. For both models, the
points were selected using a Latin Hypercube sampling (LHS) scheme \citep{mckay2000comparison}.
For each set of points, we simulated ten Ising realizations from the corresponding model
using the Swendsen-Wang algorithm with a burn-in of 100,000 samples.
The parameters were then estimated on each Ising sample, yielding ten estimates of the parameters.
The goodness-of-fit of the estimation was measured with the root mean squared error
$ \operatorname{RMSE}_i {=} \sqrt{\sum_{j=1}^{10} ( \hat\theta_{ij} {-} \theta_i)^2 / 10}$, $i{=}1,2 \ldots, 100$, as well as the
squared of the bias $\operatorname{Bias}_i^2 {=} ( \bar\theta_i {-} \theta_i)^2$, where $\bar\theta_i$ is the average value
of the ten estimates $\{\hat\theta_{ij}\}_{j=1}^{10}$.
The hundred set of points yielded overall goodness-of-fits $\operatorname{RMSE}{ =} \sum_{i=1}^{100} \operatorname{RMSE}_i / 100$, and
$\operatorname{Bias}^2 {=} \sum_{i=1}^{100} \operatorname{Bias}_i^2 /100$.

Our estimation method, which we refer to as MNG2, consists of finding the points that minimize the squared-norm of the gradient.
We do this instead of maximizing the log-likelihood because, as seen in the simulations of the previous sections,
our method approximates well the moments of the Ising model. The search for the minimum is done with the derivative-free optimization algorithm of Nelder-Mead \citep{nelder1965simplex}. The initial point of the search is sought by
evaluating a sparse grid of about fifty points.
We compare our method with the estimates given by the maximum
pseudo-likelihood estimation (MPLE)
\citep{Besag-1975,Arnold-Strauss-1991,Gong-Samaniego-1981}. 
\begin{table}
    \caption{Root means squared error (RMSE) and Bias$^2$ results for
      the (a) isotropic and (b) anisotropic model.}
    \label{table:isoaniso-estimates}
  \subfloat[istropic case]{
  \begin{tabular}{crrrr} \\ \hline
    Graph    &    \multicolumn{2}{c}{RMSE}  & \multicolumn{2}{c}{Bias$^2$} \\
  Order    &  MNG2  & MPLE   & MNG2  & MPLE    \\ \hline 
   1   & 0.67         & 0.63  & 0.39        & 0.39   \\     
   2   & 0.60         & 1.56  & 0.35        & 2.43   \\
   3   & 0.67         & 1.69  & 0.42        & 2.85   \\
   4   & 0.86         & 1.90  & 0.72        & 3.62   \\
   5   & 0.74         & 1.95  & 0.54        & 3.81   \\ \hline
  \end{tabular}}
\hfill
\subfloat[anisotropic case]{
  \begin{tabular}{ccccc} \\ \hline
  Graph  &    \multicolumn{2}{c}{RMSE}  & \multicolumn{2}{c}{Bias$^2$}  \\
  Order  &  MNG2   & MPLE   & MNG2  & MPLE \\ \hline 
   1    & 0.80         & 0.75  & 0.56        & 0.55   \\     
   2    & 0.80         & 1.58  & 0.60        & 1.52   \\
   3    & 0.80         & 1.97  & 0.62        & 3.24   \\
   4    & 0.83         & 2.10  & 0.68        & 4.64   \\
   5    & 0.85         & 2.16  & 0.73        & 4.40   \\ \hline
  \end{tabular}
}
\end{table}  
Table~\ref{table:isoaniso-estimates} has the results.
Our method is clearly more stable than MPLE over the distinct graph
models. Further, we see that 
MPLE is good at estimating first-order neighborhood parameters, but
performs poorly for higher order
neighborhood graphs.

\section{Applications}
\label{application}
\subsection{Activation detection in  fMRI experiments}
\label{fmriapp}
\subsubsection{Bayesian model for voxel activation}
We illustrate the use of our approximations in fully Bayesian
inference for determining activation in a functional Magnetic Resonance Imaging (fMRI)~\citep{lazar08,lindquist08} experiment. Our fMRI 
dataset is derived from images from the twelve replicated instances of a
subject alternating between rest and also alternately tapping his
right-hand and left-hand
fingers~\citep{maitraetal02,maitra09,maitra10}. 
For this illustration, we restrict attention only to the right-hand and 
the 20th slice, noting also that our derivations are general enough 
to extend to the other hand and  the three-dimensional volume.
Our data are in the 
form of $p$-values at each pixel that measure the significance of the
positivity of the linear relationship between the pixelwise observed
Blood-Oxygen-Level-Dependent (BOLD) time series response and the
expected BOLD response obtained through a convolution of the input
stimulus time-course with the Hemodynamic Response Function.
Let 
$p_{r1},p_{r2},\ldots,p_{rn}$ be the observed $p$-values in
the $r$th replication, where $n$ is the number of pixels. Let
$X_1,X_2,\ldots,X_n$ be indicator variables, with $X_i{=}0$ or 1
depending on whether the $i$th pixel is truly active or not. Then, we
may model $p_{ri}$, given the true state $X_i{=}x_i$ of the pixel
as $f(p_{ri}{\mid} X_i {=} x_i) {=} \{ U( p_{ri}; 0,1) \}^{1 - x_i} \,
\{b(p_{ri}; a, b)\}^{x_i}$ 
where $U( p_{ri}; 0,1)$ is the standard uniform density, and
$b(p_{ri}; a, b)$ is the density of a Beta  
distribution with parameters $(a,b)$, each evaluated at $p_{ri}$.
To simplify the analysis, we reparametrize  the
$b(p_{ri}; a, b)$ parameters by $\psi{=} a{+}b$  and
$\mu{=}a/(a{+}b)$. We assume a prior distribution on the $X_i$'s in the
form of~\eqref{eq:Ising} with homogeneous $\beta_{ij} \equiv \beta$
and a second-order neighborhood structure~\citep{li09} of $k=8$
neighbors for each interior pixel. We consider a standard uniform
prior density for $\mu$ and a $\operatorname{Gamma}( \zeta, \lambda )$
prior density for the parameter $\psi$. Specifically, the prior
density for $\psi$ is $\operatorname{gamma}(\psi; \zeta, \lambda )
{\propto} \psi^{\zeta -1} \exp( -\lambda \psi)$. We assume uniform
hyperprior densities for $\alpha$ and $\beta$. The posterior density
of $\Theta {=} \{x_1,\ldots,x_n,\alpha,\beta,\psi,\mu\}$ is given by
\begin{multline*}
  \displaystyle\prod_{i=1}^n  \displaystyle\prod_{r=1}^R \bigl\{ b(p_{ri}; \mu\psi, \psi-\mu\psi)\bigr\}^{x_i}  \tfrac{1}{Z(\alpha;\beta)}
\exp\biggl\{  -\lambda \psi + \alpha \sum_{i} x_i 
 - \displaystyle\sum_{i \sim j} \beta (1 - \delta_{ij})\biggr\} \psi^{\zeta
  -1} \mathbf{1}_{[0 < \mu < 1]},
\end{multline*}
where for any set $S$, $\mathbf{1}_{S}$ denotes the indicator function
associated with $S$. Analytical inference being impractical to implement,
we derive a Metropolis-Hastings algorithm to estimate the posterior
densities of interest. Sampling from the above needs values of
$Z(\alpha,\beta)$ for which approaches typically involve, among
other strategies, offline estimation of the normalizing constant
through tedious MCMC methods at some values and then interpolation at
others~(for example, \citep{Maitra-Besag-1998}). Our approximations
  obviate the need for this offline approach and allow for the
  possibility of a direct approach. Section~\ref{sec:appendix:1NN} provides the MCMC
  framework for each parameter   $\eta {\in} \Theta$.
\subsubsection{Results}
The  posterior density  was used in the context of activation
detection in the fMRI dataset. For this example, we set the
hyperprior
parameters $\lambda{=}1$ and $\zeta{=} 10$. This reflects our general
a priori view that $\psi{=} a {+} b$ is large. We initialized our MCMC
simulations with $(\alpha,\beta,\psi,\mu) {=} (0.001, 0.0025, 1,
0.001)$ and  collected a sample 
of 10,000 realizations after a burn-in period of the same number of
iterations. The vertex values (for the Ising model)
were updated (and 
initialized after a burn-in period of 3,000 iterations)
with the  Swendsen-Wang~\citet{Swendsen-Wang-1987} algorithm or the single-site
updating given in point (1) of Section~\ref{sec:post:densities}.
The posterior 
probability of activation, that is,  the estimated posterior means of
the $x_i$s at each voxel, are displayed in Figure~\ref{fig1}. 
As is
customary in fMRI, the image is displayed using a radiological
view. Thus, the right-hand side of the brain is imaged as the
left-hand side.
\begin{figure}[h]
  \begin{center}
    \mbox{\subfloat[]{\includegraphics[width=0.33\textwidth]{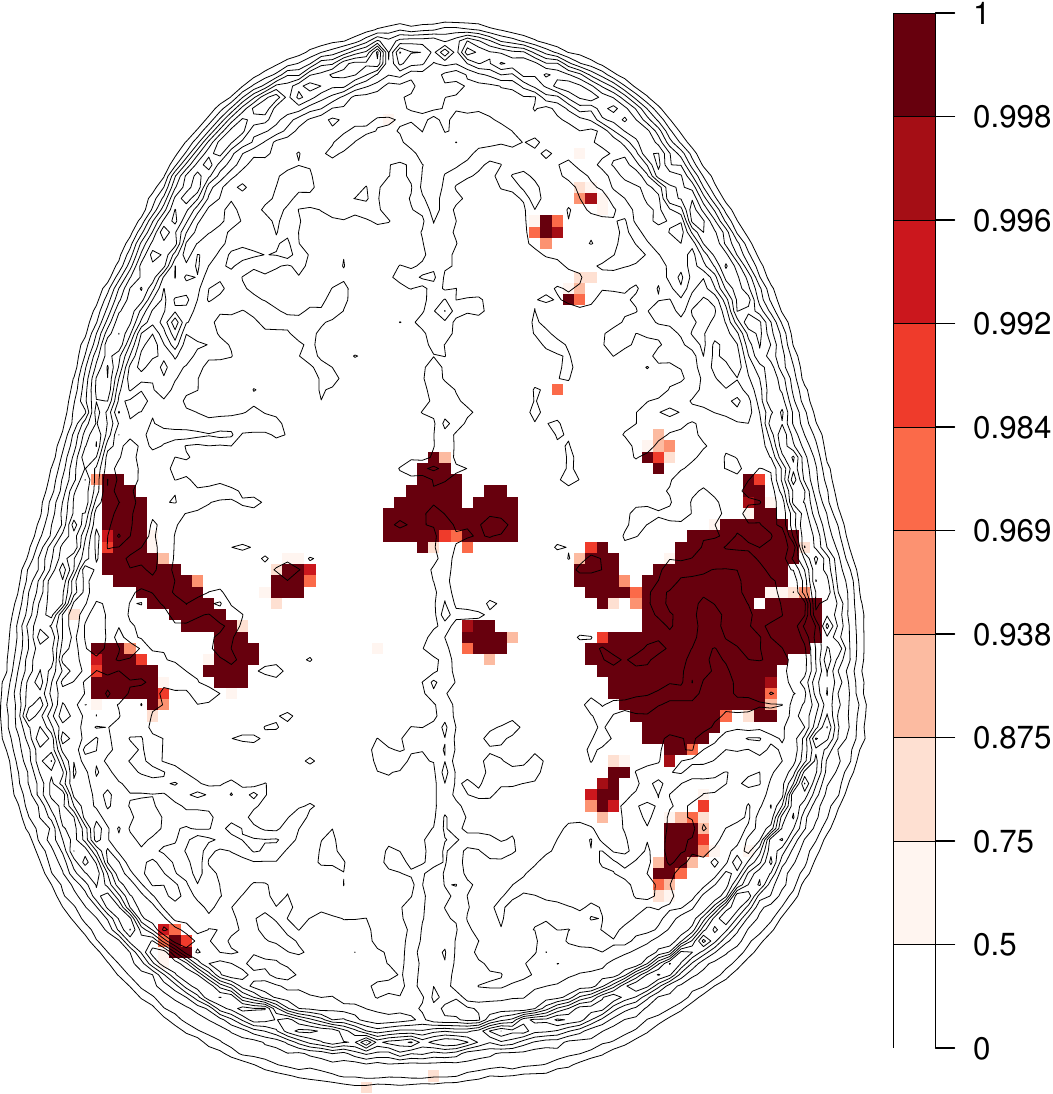}}
          \subfloat[]{\includegraphics[width=0.33\textwidth]{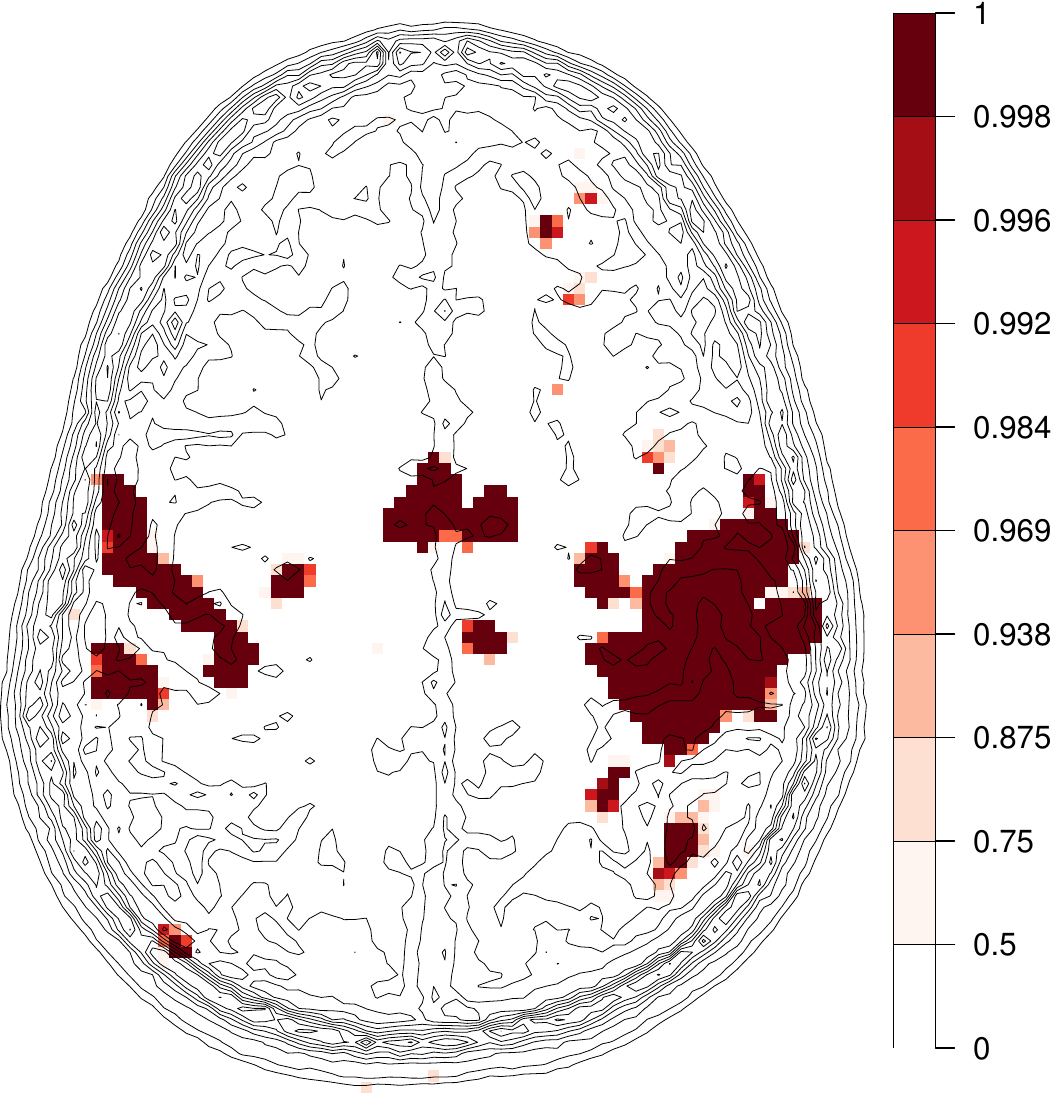}}
          \subfloat[]{\includegraphics[width=0.33\textwidth]{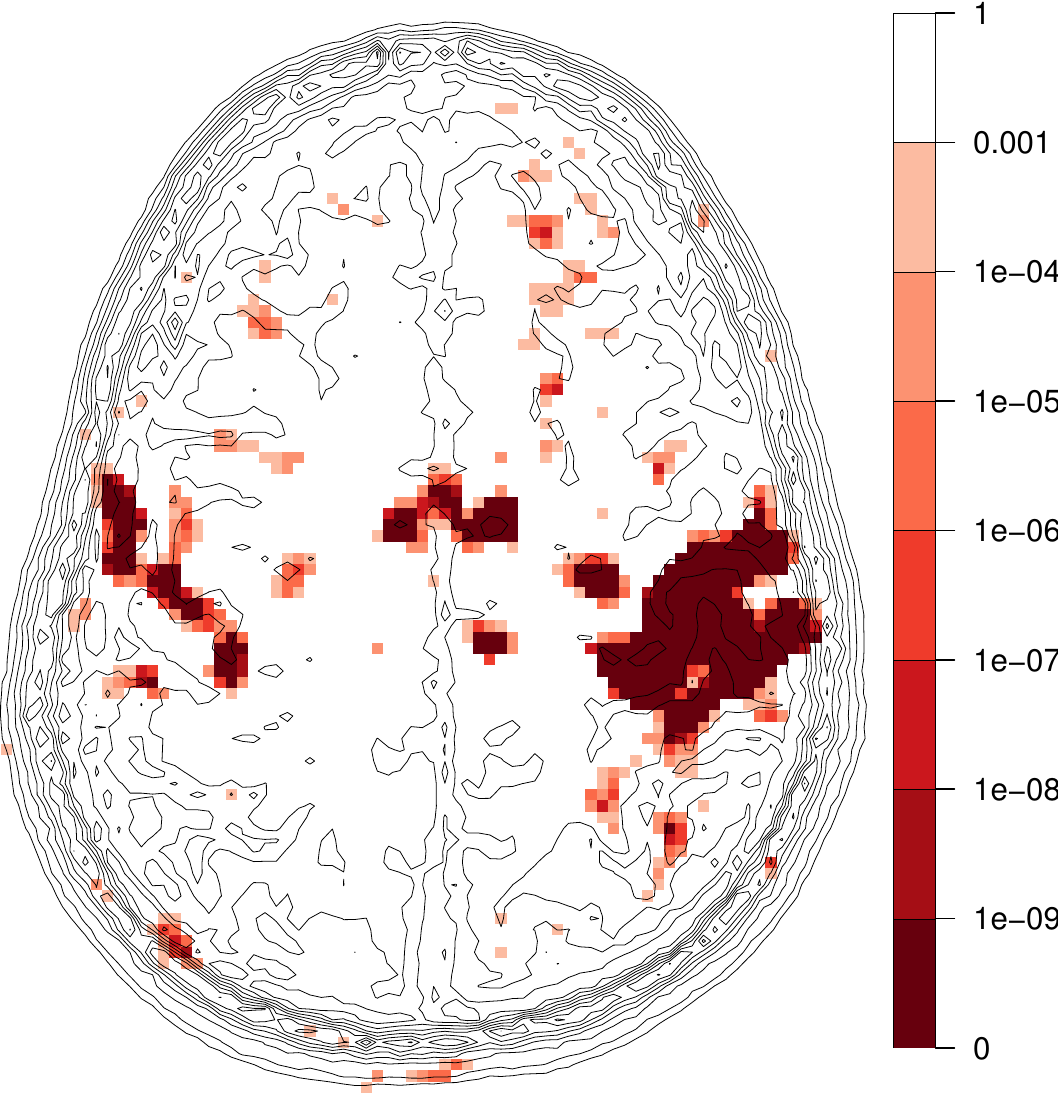}}}
  \end{center}
  \caption{
    The posterior probabilities of activation in the fMRI
    experiment using (a) Swendsen-Wang and (b) single-site
    updating. (c) Voxel-wise p-values
after cluster-wise thresholding for p-values $< 0.001$.
 Displays are in radiological view in the log scale.} 
  \label{fig1}
\end{figure}
In the figure, we only display posterior probabilities that are
greater than 0.5. 
It is clear that the
 posterior probabilities of
activation using either Swendsen-Wang or single-site updating are essentially
indistinguishable.
For comparison, we have provided the results obtained upon using the commonly-used cluster-wise thresholding of the $p$-values of the test statistic. Here, activation regions are detected by drawing clusters of connected components, each containing a pre-specified
number of voxels with $p$-values below a specified threshold~\citep{fristonetal94,hayasakaandnichols03}.  
To obtain our activation map,
we choose a 2-D second-order neighborhood, 
a threshold of 0.001  for the $p$-values, 
following the recommendations of \citet{wooetal14}, 
and a minimum cluster size of 4 pixels, as optimally recommended by \citet{fromanetal95} using the AFNI software~\citep{cox96,coxandhyde97,cox12}.
Although a detailed analysis of the results is
beyond the scope of this paper, we note from the Bayesian model
that there is very high
posterior probability of activation in the left primary motor 
(M1) and pre-motor (pre-M1) cortices and the supplementary motor
areas. There are also some areas on the right with high posterior
probability of activation, perhaps as a consequence of the left-hand
finger-tapping experiment that was also a part of the larger
experiment. 
While the activation maps using cluster-wise thresholding are generally similar to those obtained using Bayesian inference, there are many stray pixels determined to be activated. Moreover, unlike cluster-wise thresholding, Bayesian methods provide us with the posterior probability of activation and this can be used in informing further decisions.

\subsection{Establishing first order anisotropy in pistachio tree yields}
\label{sec:pistachios}
An unresolved central question in ecology is the degree to which
observed cyclic dynamics owe their spatial correlations to endogenous
or exogenous factors~\citep{nobleetal15,nobleetal18}.
Recent work~\citep{nobleetal15} demonstrated~\eqref{eq:Ising} as a
suitable model for describing long-range synchronization of
oscillations in spatial populations. In particular, the binarized yearly
increases/decreases in commercial pistachio ({\em Pistacia vera}) tree
yields, over a $66{\times}106$ grid from 2003 to 2007, shows spatial
patterns that can be fitted by an anisotropic Ising
model~\citep{nobleetal18}. However, the analysis carried out in~\citep{nobleetal18} could only
conclusively establish long-term correlations using the anisotropic
model with the binarized dataset for the annual
increases/decreases from 2003 to 2004. The authors used a
goodness-of-fit measure for 
pairwise correlation in the horizontal, vertical and diagonal
directions~\citet{nobleetal18}, with reference distributions
calculated by simulation. Our development in
Section~\ref{sec:anisotropic} enables us to easily perform a more
formal and principled likelihood ratio test (LRT), as we now demonstrate.

The data are the annual yields (rounded to the nearest pound) from 6,710 female trees at 6,996 locations on a
$66{\times}106$ grid~\citep{nobleetal18}, for the five-year period
of 2003--2007. Let $\bY_{(t)}$ be the yield of the trees over
the grid in 
year $t$, for $t\!\in\!\{2003,2004,2005,2006,2007\}$. If a tree is
male, or its yield is missing, the value at that grid-point is assumed
to be zero. Then, 
following~\citet{nobleetal18}, the first difference
field~\citep{boxandjenkins70} is $\bD_{(t)} = \bY_{(t+1)} {-}
\bY_{(t)} {-} (\bar Y_{(t+1)} {-} \bar Y_{(t)})$, where  $\bar Y_{(t)}$ is
the average of the yields in $\bY_{(t)}$. The binarized field is then
$\bX_{(t)}$ with $(i,j)$th element given by $\I(\bD_{ij,(t)} > 0)$.
\begin{figure}[h]
\vspace{-0.2in}
  \centering
  \mbox{
    \subfloat{\includegraphics[width=0.7\linewidth]{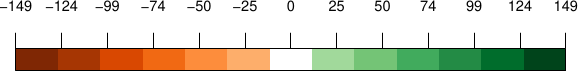}}}\setcounter{subfigure}{-0}
  \mbox{    
    \subfloat[2003-04]{\includegraphics[width=.36\linewidth]{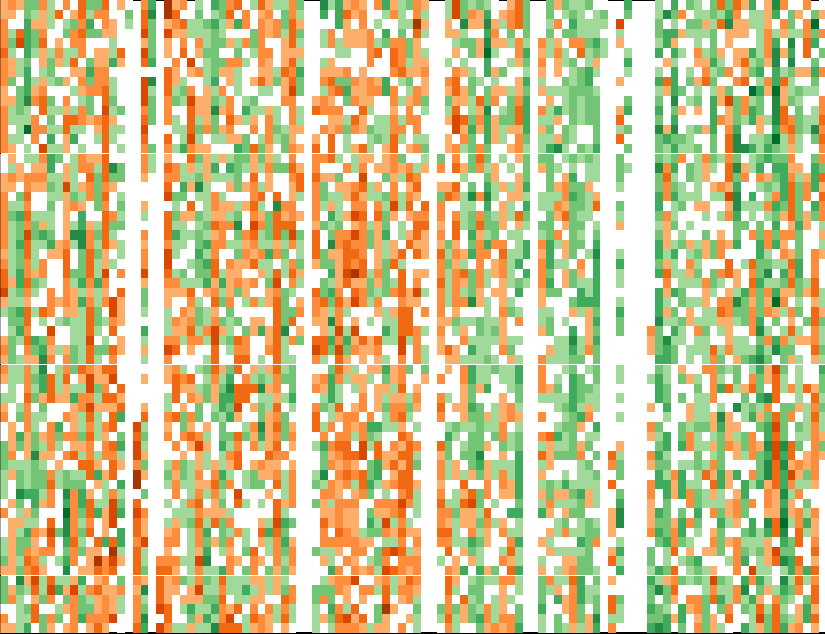}}
    \hspace{0.005\linewidth}
    \subfloat[2004-05]{\includegraphics[width=.36\linewidth]{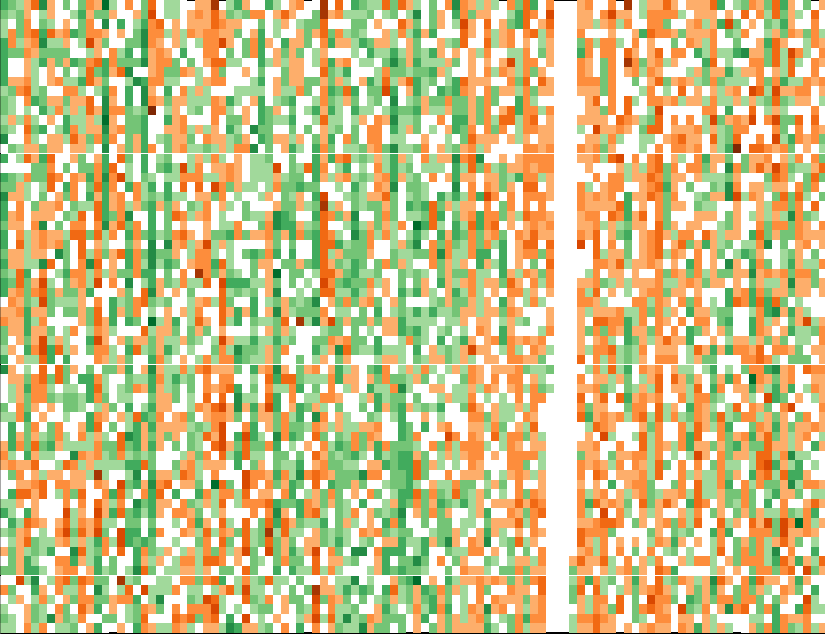}}
    \hspace{0.005\linewidth}
  }
  \mbox{
    \subfloat[2005-06]{\includegraphics[width=.36\linewidth]{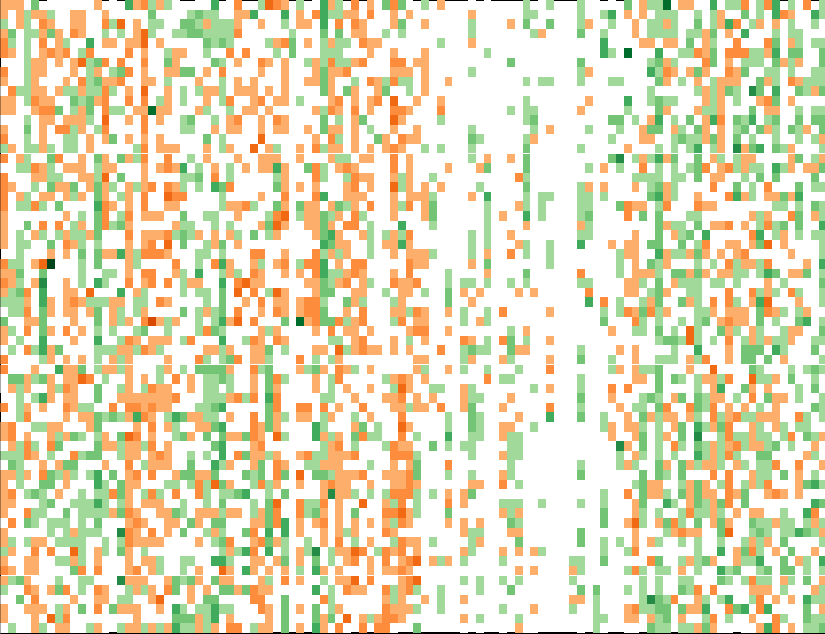}}
    \hspace{0.005\linewidth}
    \subfloat[2006-07]{\includegraphics[width=.36\linewidth]{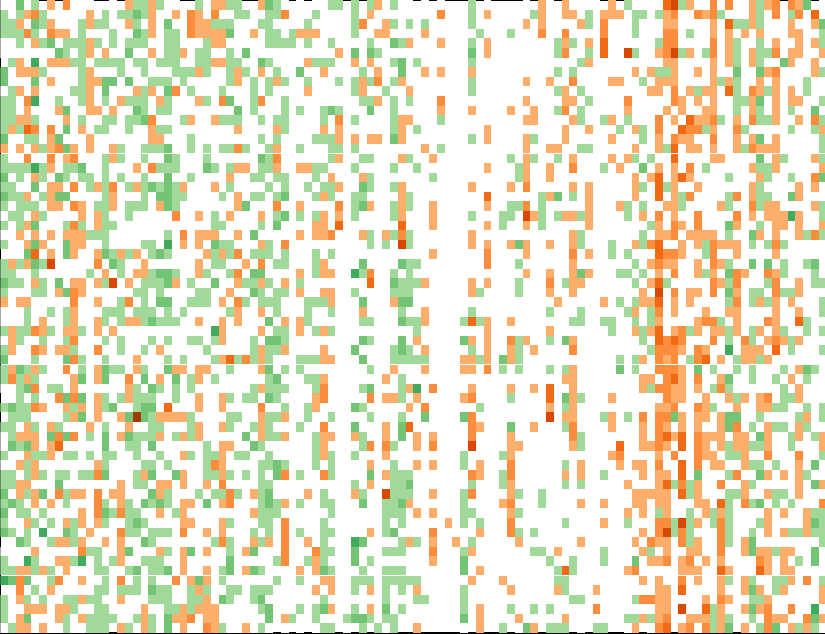}}
  }
  \caption{The differences in pistachio yield (in pounds) for female trees in the
    study area for the year (a) 2003--04, (b) 2004--05, (c) 2005--06
    and (d) 2006--07. Red colors indicate negative differences, while
    green    colors  indicate positive differences, with darker
    intensities in both cases corresponding to higher absolute
    differences. The values are zero for locations with male trees. The binarized version of these yields, that is, 
    binarized in terms of whether these differenced yields are
    positive (shaded green in each figure) or not, in each field
    is modeled by an Ising model.}
  \label{pistachiodiffs}
\end{figure}
Figure~\ref{pistachiodiffs} displays the successive differenced
yields, which were binarized for our analysis. Under the assumption of
an Ising distribution~\citep{nobleetal18}, we use a LRT to
serially test for independence,  neighborhood order, and anisotropy. 
The general form of the LRT statistic (LRTS) is the same for all these
tests, and is
\begin{equation}
  \Lambda =
  \frac{L(\widehat\alpha^0,\widehat\bbeta^0;\bx)}{L(\widehat\alpha^u,\widehat\bbeta^u;\bx)},
  \label{eq:LRTS}
\end{equation}
where $L(\alpha,\bbeta;\bx)$
is the likelihood function, and
equals the $\mathbb P(\bx;\alpha, \bbeta)$
of~\eqref{eq:pmf.iso} when the concerned hypothesis specifies an
isotropic Ising model, or is equal to the $\mathbb P(\bx;\alpha, \beta_p,
\beta_v)$ of ~\eqref{eq:pmf.aniso} when an anisotropic Ising model  is
specified by the hypothesis. Also $(\widehat\alpha^0,
\widehat\bbeta^0)$ are the maximum likelhihood estimates (MLEs)
under the null hypothesis while $(\widehat\alpha^u,\widehat\bbeta^u)$
are the MLEs under the more general model (comprising the higher
loglikelihood under the null and the alternative hypothesis specifications). We note that except
for the case of $\bbeta{=}0$, or independence between the $x_i$s, the
methods of Section~\ref{approximations} are needed to obtain
the MLE or the maximized likelihood, as well as to compute 
the $p$-value of each LRTS, which was
obtained by simulating 1,000 realizations of the field under the null
hypothesis. Indeed, even leaving aside the issue of computational costs, this is a scenario where it is not quite obvious
how one may use stochastic methods (such as MCMC) to obtain, for
instance, the partition function as a function of $\alpha$ or $\bbeta$, without recourse to estimating these quantities on a grid, and the interpolating other values. 
We now detail our investigations. 

Our first test is that of independence in the binarized
differenced yields at the sites. We consider the isotropic Ising
probability model~\eqref{eq:pmf.iso} for each
binarized differenced field $\bx$ and test  $H_0{:} \beta {=} 0$
against $H_a{:}\beta {\neq}0$, where the field is assumed
to follow~\eqref{eq:Ising}. In this case, with $\bar x\doteq
n^{-1}\sum_{i=1}^nx_i$, \eqref{eq:LRTS} has the
numerator $\bar x^{n\bar x}(1{-}\bar x)^{n-n\bar x},$
while the denominator is optimized
numerically. The $p$-value of the LRTS for each binarized differenced
yields is negligible ($p  {<} 10^{-3}$) in all four cases, and so the
null hypothesis of independence is rejected in each of these cases in
favor of the isotropic Ising model. We now evaluate anisotropy in
these binarized differenced fields.

Table~\ref{table:pistachio-tests} provides the MLEs of the parameters
upon fitting a first-order anisotropic model~\eqref{eq:pmf.aniso} to
each of the binarized differenced  
yields. In the table, $\beta_p$ is the interaction parameter in the
horizontal direction, while   $\beta_v$ is the interaction parameter in the vertical
direction.
The LRT shows significant evidence of anisotropy for each
binarized differenced field. Having established anisotropy, and in
light of the displays in Figure~\ref{pistachiodiffs} we also tested for significance of the horizontal
interaction parameter ($\beta_p$). We see
(Table~\ref{table:pistachio-tests}) that $\beta_p$ is not significant in any of the binarized
differenced fields. Finally, we note that \citet{nobleetal18} 
found
a positive $\alpha$ only in one of the four fields, while
our more formal MLEs show positive $\alpha$ estimates in both the
middle differenced 
yields. All $\alpha$s, barring the one obtained from the third
differenced field, are significant at the 5\% level of significance. 
\begin{table}[h]
  \begin{center}
  \caption{Parameter estimates for the binarized yearly differenced images of
    pistachio yields ($\bx$), along with the results of testing for
    significance. Column 5 displays the $p$-values when testing for the
    significance of $\alpha$, Column 6 for when testing first order
    isotropy against first order anisotropy, while Column 7 provides
    the $p$-values for when testing for a horizontal interaction
    parameter in the first order anisotropic model.}
    \label{table:pistachio-tests}
    \begin{tabular}{c|clc|ccc} \hline
      &   \multicolumn{3}{c|}{Parameters}
      &\multicolumn{3}{c}{$p$-value of hypothesis tests}

                                              \\     
   $\bx$  &  $\alpha$    &  $\beta_v$  & $\beta_p$ & $H_0{:}\alpha{=}0$ &
                                                           $H_0{:}\beta_p{=}\beta_v$
                                                   &
                                                     $H_0{:}\beta_p{=}0$\\ \hline
      I      &    -0.100   &  0.908    & 0.319   & 0.025 & $ < 10^{-3}$   &  0.395 \\
      II     &    0.035   &  0.967  & 0.315   &  0.028 & $<10^{-3}$     &  0.823   \\
      III    &    0.036   &  0.948    & 0.353   & 0.091 & 0.002           &  0.498  \\   
   IV     &    -0.069   &  0.823    & 0.374  &   0.026 & $ < 10^{-3}$  &  0.951  \\ \hline   
  \end{tabular}
\end{center}
\end{table}
The results of our likelihood-based approach extend the findings
of~\citet{nobleetal18}, by establishing anisotropy in all the binarized
differenced yields. Specifically, previous work~\citep{nobleetal18} was able to
conclusively establish 
anisotropy only for the 2003-04 differenced yields: our approach shows
that a first order anisotropic Ising model provides a significantly
better fit in all four cases. Further, we show that there is no
significant interaction effect in the  horizontal direction in either of
these binarized differenced yields.

\section{Discussion}
\label{sec:discussion}
In this paper, we provided explicit approximations for difficult-to-calculate quantities derived from the homogeneous Ising model. In particular, 
we developed approximations to the partition function, and the moments that are otherwise
intractable even for moderate-sized graphs.
An R~\citep{R} package implementing our approximation formulae is under
development and will be released to accompany this paper.
We showed that our approximation works well for very realistic 
lattice sizes and neighborihood structures in the lattices.
We  stress that our approximations apply to general regular
graphs, not necessarily lattices,
and with general neighborhood structure. Indeed, its derivation
does not use any special graph structure, and  only supposes that the graph
is regular, that is, that the degree $k$ of each vertex is the same for
the entire graph. Recall from the introduction 
that many researchers have had to simplify their models because of the
difficulty in obtaining quantities such as the partition function of
more realistic models. 
Further, unlike the restrictive special cases
(such as those considered  by \citet{kaufman49,schraudolphandkamenetsky09,karandashevandmalsagov17})
where calculation of  the exact partition function is possible, our
approach does not grow with the size of the graph and can apply to all
dimensions as long as the graph structure is regular. 
Therefore we expect that our approximation will
facilitate   the Ising modeling of complex systems by allowing fast and reliable
inference. An example of such a situation is a fully Bayesian approach
to activation detection in fMRI, which we have demonstrated in
Section~\ref{fmriapp} can be done quite speedily using our
approximations. (This is because our method provides formulae for $Z(\alpha,\beta)$, which would otherwise have to be estimated separately and individually, using stochastic methods for each combination of $\alpha,\beta$ as needed in the MCMC. In essence therefore, our approach when used in fully Bayesian estimation involving a homogeneous Ising model prior,  obviates the need for a second layer of MCMC to estimate $Z(\alpha,\beta)$ for each $(\alpha,\beta)$, and therefore can be used in conjunction with the fastest of stochastic algorithms to further speed up computations.) 
A second application of our methodology, demonstrated in Section~\ref{sec:pistachios},
establishes anisotropy in differenced pistachio tree yields using a
formal LRT framework. The LRTS requires ML estimation of the parameters
under the null and the more general hypothesis of the null or the
alternate hypotheses, and it is our approximate but accurate formulae
developed here that make it possible to implement a formal LRT in this
application.  

A few more comments are in order. 
The  
normal approximations of Section~\ref{approximations} can be
computed in ${\cal O}(n)$ operations. However, further approximations
obtained by replacing the sums by integrals may
be computed much faster depending on the integration method used.
We note that the analytical estimates for the mean activation and spin
interaction work better for large $k$ because the normal 
approximation of the distribution of the mean of  $r_h {=} {\sum_j} \eta_{i_h j} x_j$
is more suited for large $k$.

A possible extension pertains to the case of approximations for the
nonhomogeneous Ising model, which is important in a number of applications, such as in item response theory~\citep{vanborkuloetal14}. We feel that it is generally difficult to specify approximation formulae for nonhomogeneous Ising model in the abstract, and that this will depend on the exact specification of the model. However, we also believe that our approximations in this paper 
will be straightforward to generalize to a locally varying interaction parameter $\beta_{ij}$ if 
we could somehow decompose the graph into regular components (homogeneous regions
with approximately constant degree) where the $\beta_{ij}$s do not
change much. 
In this case, the error in the approximation would probably depend on 
the size of each component, say $n_c$ which are necessarily  assumed
to be  large for our potential approximations to hold.

We close this section with one last comment on the fMRI
application. One aspect that has so far not been invoked in fMRI is
the fact that it is known that only a very small proportion of about 0.5-2\% of voxels are activated in a typical fMRI
study~\citep{lazar08,chenandmaitra23}. However, this information has never been
incorporated satisfactorily in the context of Bayesian activation
detection of fMRI. Our approximations in Section~\ref{approx:M} make
it possible to perform Bayesian inference while constraining the prior
parameters $(\alpha,\beta)$ so that the {a priori} proportion of
expected activated voxels is satisfied. Thorough development and
implementation of such methodology for this application would be of
great practical interest for reliably assessing cognition. Thus, we
see that while we have addressed an important problem in this paper,
there remain issues meriting further attention.

\ifCLASSOPTIONcaptionsoff
  \newpage
\fi

\renewcommand\thefigure{S\arabic{figure}}\setcounter{figure}{0}
\renewcommand\thetable{S\arabic{table}}\setcounter{table}{0}
\renewcommand\thesection{S\arabic{section}}\setcounter{section}{0}
\renewcommand\theequation{S\arabic{equation}}\setcounter{equation}{0}
\section*{Supplementary Materials}
\section{Supplement to Section~\ref{approximations}}
\label{supp:approximations}
The approximation of $\E(M)$ is
\begin{equation*}
  \begin{split}
  \tilde{M}_{\phi}(\alpha, \beta_p, \beta_v) & = \tfrac{1}{Z_{\tilde\phi}(\alpha, \beta_p, \beta_v) }
  \tfrac{\partial}{\partial \alpha }Z_{\tilde\phi}(\alpha, \beta_p, \beta_v) \\
 & =
   \tfrac{1}{Z_{\tilde\phi}(\alpha, \beta_p, \beta_v) } \biggl[ n
   \exp{(\alpha n)} + n \exp{(\alpha'')} \\
& \qquad\qquad\qquad\qquad + n(n-1)  \exp{\{\alpha (n-2) + \alpha''\}} 
   + C_{2, \phi}(\alpha, \beta_p, \beta_v) 
    + n\sqrt{\tfrac{n}{2\pi}}    \int_{2/n}^{1 - 2/n}
    \tfrac{\Delta_{\Phi_2}(ny) \exp\{ g_2(ny)\} }{
    (1-y)^{n(1-y) + \frac{1}{2}} y^{ny - \frac{1}{2} }}
    dy \biggr],
  \end{split}
\end{equation*}
where
$C_{2, \phi}(\alpha, \beta_p, \beta_v) ={n \choose 2} \bigl(\tfrac{n-2}{2}
 \exp\{g_2(n-2)\}\Delta_{\Phi_2}(n-2) + \exp\{g_2(2)\}\Delta_{\Phi_2}(2) \bigr)$.
 The approximation of $\E(S_p)$ is
 \begin{multline*}
   S_{p, \tilde\phi}(\alpha, \beta_p, \beta_v) =
  \tfrac{(n k_p/2)}{Z_{\tilde\phi}(\alpha, \beta_p, \beta_v) } \biggl[
  \exp{(\alpha n)} +  (n-2) \exp{\{\alpha (n-2) + \alpha''\}} + D_{2, p, \phi}(\alpha, \beta_p, \beta_v) 
    + \sqrt{\tfrac{n}{2\pi}}    \int_{2/n}^{1 - 2/n}
    \tfrac{ \tilde{H}_p(ny) \exp\{ g_2(ny)\} }{
    (1-y)^{n(1-y) + \frac{1}{2}} y^{ny - \frac{3}{2} }}
    dy \biggr], 
\end{multline*}
 where
 $ n k_p D_{2, p, \phi}(\alpha, \beta_p, \beta_v) =  {n\choose 2} \big[ H_p(n-2) \exp\{g_2(n-2)\} + H_p(2)\exp\{g_2(2)\}\big]$,
 with $\tilde{H}_p(\ell) = H_p(\ell)/( 2\ell_2\theta_p)$, and
 $H_p(\ell)  = \bigl\{ \Delta_{\Phi_2}(\ell) g_{2, p}(\ell) + \Delta_{\Phi_2, p}(\ell) \bigr\}, $
$ g_{2, p}(\ell) = \mu_{p, \ell} + \beta_p \tau_{p,\ell}^2 + \beta_v \rho_{pv,\ell}\tau_{p,\ell} \tau_{v,\ell}, $,
   and  $$\Delta_{\Phi_2, p}(\ell) = - \tau_{p,\ell} \biggl\{ \phi\bigl( \tfrac{u^\bullet}{\tau_{p,\ell}}\bigr) \Delta_{2,\Phi}(u^\bullet) 
 -   \phi\bigl( \tfrac{u_\bullet}{\tau_{p,\ell}}\bigr) \Delta_{2, \Phi}(u_\bullet) \biggr\} 
 - \tfrac{\rho_{pv,\ell}}{\sqrt{1 - \rho_{pv, \ell}^2}} \int_{u_\bullet}^{u^\bullet} \phi\bigl(\tfrac{u}{\tau_{p,\ell}}\bigr) \Delta_{2,\phi}(u)
 \, du,
$$
 with
$\Delta_{2,\Phi}(u) = 
\Phi\biggl( \tfrac{ v^{\bullet} - u\rho_{pv, \ell}\tfrac{\tau_{v,\ell}}{\tau_{p,\ell}}}{\tau_{v,\ell} \sqrt{ 1 - \rho_{pv, \ell}^2} } \biggr)
- \Phi\biggl( \tfrac{ v_{\bullet} - u\rho_{pv, \ell}\tfrac{\tau_{v,\ell}}{\tau_{p,\ell}}}{\tau_{v,\ell} \sqrt{ 1 - \rho_{pv, \ell}^2} } \biggr),$ and
$\Delta_{2,\phi}(u) =
\phi\biggl( \tfrac{ v^{\bullet} - u\rho_{pv, \ell}\tfrac{\tau_{v,\ell}}{\tau_{p,\ell}}}{\tau_{v,\ell} \sqrt{ 1 - \rho_{pv, \ell}^2} } \biggr)
      - \phi\biggl( \tfrac{ v_{\bullet} - u\rho_{pv, \ell}\tfrac{\tau_{v,\ell}}{\tau_{p,\ell}}}{\tau_{v,\ell} \sqrt{ 1 - \rho_{pv, \ell}^2} } \biggr).$

Proceeding in exactly the same manner, we have the approximation of  $\E(S_v)$:
\begin{equation*}
  \begin{split}
  S_{v, \tilde\phi}(\alpha, \beta_p, \beta_v) & = \tfrac{1}{2 Z_{\tilde\phi}(\alpha, \beta_p, \beta_v) }
  \tfrac{\partial}{\partial \beta_v }Z_{\tilde\phi}(\alpha, \beta_p,
                                                \beta_v) \\   & =
  \tfrac{1}{Z_{\tilde\phi}(\alpha, \beta_p, \beta_v) } \biggl\{
  \tfrac{n k_v}{2} \exp(\alpha n) + D_{2, v, \phi}(\alpha, \beta_p, \beta_v)  
    + \tfrac{1}{2} \sqrt{\tfrac{n}{2\pi}}    \int_{2/n}^{1 - 1/n}
    \tfrac{H_v(ny) \exp\{ g_2(ny)\} }{
    (1-y)^{n(1-y) + \frac{1}{2}} y^{ny + \frac{1}{2} }},
    dy \biggr\},
    \end{split}
 \end{equation*}
 where
 $2 D_{2, v, \phi}(\alpha, \beta_p, \beta_v) = n H_v(n-1) \exp\{g_2(n-1)\} + {n \choose 2}H_v(2)\exp\{g_2(2)\}$,
 with 
 $H_v(\ell)  = \bigl\{ \Delta_{\Phi_2}(\ell) g_{2, v}(\ell) + \Delta_{\Phi_2, v}(\ell) \bigr\}, $
 where
 $g_{2, v}(\ell) = \mu_{v, \ell} + \beta_v \tau_{v,\ell}^2 + \beta_p \rho_{pv,\ell}\tau_{p,\ell} \tau_{v,\ell},$
 and
 $$ \Delta_{\Phi_2, v}(\ell)=  -
 \tau_{v,\ell} \biggl\{ \phi\bigl( \tfrac{v^\bullet}{\tau_{v,\ell}}\bigr) \tilde\Delta_{2,\Phi}(v^\bullet) 
 -   \phi\bigl( \tfrac{v_\bullet}{\tau_{v,\ell}}\bigr) \tilde\Delta_{2, \Phi}(v_\bullet) \biggr\} 
 - \tfrac{\rho_{pv,\ell}}{\sqrt{1 - \rho_{pv, \ell}^2}} \int_{v_\bullet}^{v^\bullet} \phi\bigl(\tfrac{v}{\tau_{v,\ell}}\bigr) \tilde\Delta_{2,\phi}(v)
 \, dv,$$
 where
$\tilde\Delta_{2,\Phi}(v) = 
\Phi\biggl( \tfrac{ u^{\bullet} - v\rho_{pv, \ell}\tfrac{\tau_{p,\ell}}{\tau_{v,\ell}}}{\tau_{p,\ell} \sqrt{ 1 - \rho_{pv, \ell}^2} } \biggr)
- \Phi\biggl( \tfrac{ u_{\bullet} - v\rho_{pv, \ell}\tfrac{\tau_{p,\ell}}{\tau_{v,\ell}}}{\tau_{p,\ell} \sqrt{ 1 - \rho_{pv, \ell}^2} } \biggr),$ and
$\tilde\Delta_{2,\phi}(u) =
\phi\biggl( \tfrac{ u^{\bullet} - v\rho_{pv, \ell}\tfrac{\tau_{p,\ell}}{\tau_{v,\ell}}}{\tau_{p,\ell} \sqrt{ 1 - \rho_{pv, \ell}^2} } \biggr)
      - \phi\biggl( \tfrac{ u_{\bullet} - v\rho_{pv,
          \ell}\tfrac{\tau_{p,\ell}}{\tau_{v,\ell}}}{\tau_{p,\ell}
        \sqrt{ 1 - \rho_{pv, \ell}^2} } \biggr).$

\section{Supplement to Section~\ref{simulations}}
\subsection{Normalizing constant for a 2-NN graph}
\label{sec:appendix:1NN}
Consider the circular 2-NN Ising model with variables
$\{ v_1,\ldots, v_n\} \subset \{ -1,1\}$ (in this case,
$v_n$ is a neighbor of $v_1$).
The corresponding normalizing constant is given by
$$Z_{1}(L, K) = \displaystyle\sum_{\{v\}} \exp\{ L \displaystyle\sum_{i} v_i +  K (v_nv_1 + \displaystyle\sum_{i=1}^{n-1} v_i v_{i+1}) \}.$$
A well-established result \citep[Chapter 13, p. 261]{Salinas-2001} says that for large $n$,
\begin{multline}
\log Z_{1} (L, K) {\approx} n \log\left\{  \exp(K ) \cosh(L){+} \sqrt{ \exp(2K ) \cosh^2(L) {-} 2\sinh(2K)} \right\}.
\label{eq:z1}
\end{multline}
In our setup, the variables $x_i \in \{0,1\}$. So we need to transform
the variables $v_i$ to $x_i$. It is easy to see that the corresponding
transformation is $x_i = (1 + v_i)/ 2$.
so that
\begin{equation*}
  \begin{split}
    Z_{1}(L, K) & = \displaystyle\sum_{\{x\}} \exp\left[ L \displaystyle\sum_{i} (2x_i-1) +  K \bigg\{ (2x_n-1)(2x_1-1) + \displaystyle\sum_{i=1}^{n-1} (2x_i-1)(2 x_{i+1}-1 ) \bigg\} \right]  \\
                & =  \exp[ n( K - L) ]                                                                               \displaystyle\sum_{\{x\}} \exp\left[
2(L-2K) \displaystyle\sum_{i} x_i 
+ 4K ( x_nx_1 + \displaystyle\sum_{i=1}^{n-1} x_ix_{i+1} )  \right].\\
    \end{split}
  \end{equation*}

From here, we get that
$Z_1( K, L) = Z( 2(L - 2K), 2K ) \exp\{n(K-L)\}$,
where the $4K$ is replaced by $2K$ because in our model the sum of over the
neighboring vertices is multiplied by two.
Therefore, using the fact that the one-nearest-neighbor graph is a regular graph with $k=2$, we get
$K = \beta/2$, and $L= \alpha/2$.
Consequently,
$Z(\alpha, \beta) = \exp\{ (\alpha - \beta) (n /2 ) \} Z_1( \alpha/2, \beta/2 ).$
Using \eqref{eq:z1}, we obtain for large $n$,
\begin{multline*}
  {\log Z(\alpha, \beta) 
 \approx
\frac {n(\alpha - \beta )}2 
+ n\log \biggl[ \exp\left(\tfrac\beta2 \right) \cosh( \alpha/2) + \sqrt{ \exp(\beta ) \cosh^2\left(\tfrac\alpha2\right) - 2\sinh(\beta)}\biggr].}
\end{multline*}

      \section{Supplement to Section~\ref{application}}
      
\subsection{Posterior densities for MCMC simulations}
\label{sec:post:densities}
  \begin{enumerate}
\item {\bf The parameter $x_i$'s}: Let $x_{(-i)}=\{x_1, \ldots,  x_n\}\setminus \{x _i \}$. We propose to update  $x_i\in \{0, 1\}$
  via Gibbs' sampling. The full conditional of the posterior
  distribution of $x_i$, $p(x_i\mid x_{(-i)},\Theta_{-x},p)$ is proportional to
\begin{equation*}
  \biggl\{ \bigl(
    C_B\bigr)^R  \prod_{r=1}^R  p_{ri}^{\mu\psi} (1
    -p_{ri})^{(1-\mu)\psi} \biggr\}^{x_i}  \exp\{ \alpha x_i - \beta \sum_{j\sim i} (1 -\delta_{ij}) \},
\end{equation*}
with
$C_B = \Gamma(\psi)/ \{ \Gamma(\mu\psi) \Gamma((1-\mu)\psi)\}$.
Let $A_i = \alpha + (\mu\psi -1) \sum_{r=1}^R \log p_{ri} + \{(1-\mu)\psi -1\}
\sum_{r=1}^R \log ( 1- p_{ri}) + R\log C_B.$
Then $\sum_{r=1}^R \log p_{ri}$ is
the same as $R\log \check p_i$ where  $\check p_i$ is the  harmonic mean of
$\{p_{ri}:r=1,2,\ldots,R\}$. Also, $\sum_{r=1}^R \log ( 1-
p_{ri})/R\equiv  \log\check q_i$  where $\check q_i$ is the harmonic mean of
$\{1-p_{ri}:r=1,2,\ldots,R\}$. Then 
$A_i \equiv \alpha + (\mu\psi -1) R \log( \check{p}_i ) + \{(1-\mu)\psi -1\}
R \log( \check{q}_i )  + R\log C_B$
and 
\begin{equation*}
\operatorname{Pr}( X_i= 1 | x_{(-i)},\Theta_{-x},p) \propto 
\tfrac{ \exp( A_i - \sum_{j\sim i}\beta) }
{\exp( A_i  - \sum_{j\sim i}\beta     )
+ \exp\{ - 2(\beta \sum_{j\sim i} x_{j}) \}  }
\end{equation*}

\item {\bf The parameter  $\beta$:} We have
$f( \beta | \Theta_{-\beta}, p) \propto { \exp\{ - \beta \sum_{i \sim j} (1 -\delta_{ij}) \} }/{Z( \alpha, \beta)}$,
  where $Z( \alpha, \beta)$
  will be
approximated numerically 
by $\tilde{Z}_\phi( \alpha, \beta)$ as defined in Section~\ref{approx:Z}.
We let $\ddot\beta\sim$Gamma$(a,b)$ as our proposed update to $\beta$ with
$a/b = \beta$, and $a /b^2 = \gamma$, for some moderate $\gamma>0$
(e.g., $\gamma = 1$).  
This yields the \citet{hastings70} acceptance ratio that is the
minimum of 1 and 
\begin{equation*}
\tfrac{\tilde Z_\phi(\alpha,\beta)}{\tilde Z_\phi(\alpha,\ddot\beta)}
\tfrac{\Gamma ({\beta^2}/{\gamma})}{\Gamma({\ddot\beta^2}/{\gamma} )} 
\exp\biggl\{ -(\ddot\beta - \beta) \sum_{i\sim j} (1 - \delta_{ij})  
+ \frac1\gamma( \ddot\beta^2 - \beta^2)
\log(\beta\ddot\beta/\gamma)  \biggr\}.
\end{equation*}
\item {\bf The parameter  $\alpha$:} The full conditional for $\alpha$
  is given by
\[f( \alpha | \Theta_{-\alpha}, p) \propto
 \exp(  \alpha \sum_{i=1}^n x_{i} )/ Z( \alpha, \beta).\]
 Our proposed update to $\alpha$ is $\ddot\alpha\sim
N(\alpha, \sigma_\alpha^2)$ with a moderate $\sigma^2_\alpha$, (e.g.,
$\sigma^2_\alpha = 1$), which, after incorporating the approximation in
Section~\ref{approx:Z}.
is accepted in a~\citep{metropolisetal53}
step with probability given by
$\min\left\{1,
{\tilde Z_\phi(\alpha,\beta)}\exp\{ (\ddot\alpha - \alpha) \sum_{i=1}^n x_{i} \} /{\tilde Z_\phi(\ddot\alpha,\beta)}\right\}$.

\item {\bf The parameter  $\psi$:} The full conditional for $\psi$ is given by
  \begin{equation*}
  f( \psi | \Theta_{-\psi}, p ) \propto 
\biggl[ \tfrac{\Gamma(\psi)}{ \Gamma(\mu\psi) \Gamma([1-\mu]\psi)}
  \biggr]^{R \sum_{i}\! x_i} \psi^{\zeta - 1} \exp[ \psi \{ \mu\log
  A + (1-\mu) \log B  - \lambda \} ],
\end{equation*}
where $A= \prod_{i=1}^n \check{p}_i^{R x_i}$,
and $B = \prod_{i=1}^n \check{q}_i^{R x_i}$.
Using \citet{stirling1730}'s approximation to the factorial function,
the right-hand-side of the above equation can be approximated by
\begin{equation}
  \psi^{\frac{1}{2}n_1 + \zeta -1} \! \exp\biggl[  -\psi \{ \lambda + \operatorname{Ent}(\mu) 
    - \mu\log A - (1-\mu) \log B  \} \biggr],
\label{eq:psi}
\end{equation}
where $n_1= R \sum_{i=1}^n x_i$, and 
 $\operatorname{Ent}(\mu)= \mu \log \mu + (1-\mu) \log(1-\mu)$ is the negative of the entropy associated with the probabilities
$(\mu,  1 -\mu)$.
We update $\psi$ with the proposal
$\ddot\psi\sim$Gamma$( \tfrac{1}{2}n_1 + \zeta,
\lambda + \operatorname{Ent}(\mu) - \mu\log A - (1-\mu) \log B ).$ 
More generally, we can reduce this updating step to a Gibbs'
sampling step using the above approximated full conditional. 
The acceptance probability for $\ddot\psi$ is the minimum of 1 and 
\begin{multline*}
  \left(\tfrac{\Gamma(\ddot\psi)\, \Gamma(\mu\psi) \Gamma((1-\mu)\psi)}
       {\Gamma(\mu\ddot\psi) \Gamma((1-\mu)\ddot\psi)\,
            \Gamma(\psi)} \right)^{R \sum_{i}\! x_i}  \biggl(\tfrac{\psi}{\ddot\psi}\biggr)^{\frac{1}{2}n_1 + \zeta -1} \biggl(\tfrac{\ddot\psi}{\psi}\biggr)^{\zeta - 1} \!\!\!\!
\exp\bigl\{ (\ddot\psi -\psi)( \mu\log A + (1-\mu) \log B  - \lambda ) \bigr\} \biggr. \\
\times \exp\bigl[ -  \{ \lambda - \operatorname{Ent}(\mu) - \mu\log A -
(1-\mu) \log B  \} (\psi - \ddot\psi) \bigr],
\end{multline*}
which using~\eqref{eq:psi} yields the approximated acceptance ratio 
$
\min\{1, 
\exp[ (\psi - \ddot\psi) n_1 \{ \mu \log\mu + (1-\mu) \log( 1 - \mu)
\} ] \}$,
which is 1 whenever $\ddot\psi > \psi$ and so proposals $\ddot\psi >
\psi$ are always accepted.

\item {\bf The parameter  $\mu$:} Here, we use a random walk update
  $\ddot\mu \sim \operatorname{U}(0,1)$. This yields the 
  Metropolis-Hastings acceptance ratio that is the minimum of 1 and 
\begin{equation*}
\biggl(\tfrac{\Gamma(\psi)}{\Gamma(\ddot\mu\psi) \Gamma((1-\ddot\mu)\psi)}
\tfrac{\Gamma(\mu\psi) \Gamma((1-\mu)\psi)}{\Gamma(\psi)} \biggr)^{n_1} \, \\
\exp\biggl\{ \psi \ddot\mu \log A + \psi (1 -\ddot\mu) \log B \\
- \psi \mu \log A - \psi (1-\mu)  \log B \biggr\}.
\end{equation*}
\end{enumerate}

\bibliographystyle{IEEEtran}
\bibliography{references_fMRI,references_appl}

\end{document}